\documentclass[11pt, reqno]{amsart}
\usepackage{graphicx, epsfig, psfrag}
\usepackage{amsfonts}
\usepackage{amssymb}
\usepackage{amsmath}
\usepackage{amsaddr}
\usepackage{float}
\usepackage{hyperref}
\usepackage{pdfsync}
\usepackage{xcolor}
\usepackage{tikz}
\usetikzlibrary{shapes.misc, positioning}
\usetikzlibrary{patterns,snakes}
\usepackage{algorithm}
\usepackage{algpseudocode}
\usepackage{textcomp}
\usepackage{listings}
\usepackage[sort]{cite}

\allowdisplaybreaks[4]

\def\Ls{\mathcal{L}}
\def\sb{\boldsymbol{s}}
\def\tr{\text{tr}}
\def\ee{\text{e}}
\def\ii{\text{i}}
\def\dd{\text{d}}
\def\sgn{\text{sgn}}
\def\bB{\bar{B}}
\def\C{\mathbb{C}}

\def\P{\mathbb{P}}
\def\tP{\tilde{\mathbb{P}}}

\newcommand{\mf}[1]{\mathfrak{#1}}
\newcommand{\mc}[1]{\mathcal{#1}}
\newcommand\mQ{\mathcal{Q}}
\newcommand{\redtext}[1]{{\color{red}#1}}
\newcommand{\bluetext}[1]{{\color{blue}#1}}
\newcommand{\greentext}[1]{{\color{green}#1}}

\newcommand\Sf{s_{\mathrm{f}}}
\newcommand\Si{s_{\mathrm{i}}}

\newtheorem{theorem}{Theorem}
\newtheorem{lemma}{Lemma}
\newtheorem{definition}{Definition}
{\theoremstyle{remark} }

\title{Inclusion-Exclusion Principle for Open Quantum Systems with Bosonic Bath} 
\thanks{Funding: ZC~was supported by the Academic Research Fund of the Ministry of Education of Singapore under grant No. R-146-000-291-114. The work of JL~was supported in part by the National Science Foundation via grants  DMS-2012286 and CHE-2037263.}

\author{Siyao Yang, Zhenning Cai}
\address[SY, ZC]{Department of Mathematics, National University of
  Singapore, Level 4, Block S17, 10 Lower Kent Ridge Road, Singapore 119076.}
\email{matsiya@nus.edu.sg}
\email{matcz@nus.edu.sg}

\author{Jianfeng Lu}
\address[JL]{Department of Mathematics, Department of Physics, and
  Department of Chemistry \\ Duke University, Box 90320, Durham NC 27708, USA.}
\email{jianfeng@math.duke.edu}

\begin{document}

\begin{abstract}
    We present two fast algorithms which apply inclusion-exclusion principle to sum over the bosonic diagrams in bare diagrammatic quantum Monte Carlo (dQMC) and inchworm Monte Carlo method, respectively. In the case of inchworm Monte Carlo, the proposed fast algorithm gives an extension to the work [``Inclusion-exclusion principle for many-body diagrammatics'', Phys.~Rev.~B, 98:115152, 2018] from fermionic to bosonic systems. We prove that the proposed fast algorithms reduce the computational complexity from double factorial to exponential. Numerical experiments are carried out to verify the theoretical results and to compare the efficiency of the methods.
\end{abstract}

\date{\today}

 \maketitle

\section{Introduction}

Open quantum systems, which characterize quantum systems coupled with environment, 
have been studied extensively for many decades, as it arises in many context including quantum optics \cite{Breuer2007}, quantum computation \cite{Nielsen2010}, and dynamical mean field theory  \cite{Gull2011RMP}, just to list a few. The coupling between the system and the environment leads to non-Markovian evolution of the quantum state of the system. In the weak coupling limit, such evolution can be approximated by the Markovian process described by the Lindblad equation \cite{Davies1974, Davies1976}, which simplifies the numerical simulation. In the more challenging case where memory effect has to be taken into account, a number of numerical methods have been proposed in the literature. For example, the quasi-adiabatic propagator
path integral (QuAPI) \cite{Makri1995, Makri1996} method assumes finite memory length and so that the path integral can be numerically computed iteratively; by assuming that the bath response function has a special form, the hierarchical equations of motion can be applied \cite{Tanimura1989,Tanimura1990}; the method of multiconfiguration time dependent Hartree (MCTDH) \cite{Beck2000} is developed based on ansatz of wave functions. While these deterministic methods require some additional modeling of the open quantum system, the bare diagrammatic quantum Monte Carlo (dQMC) method \cite{Keldysh1965} applies Monte Carlo sampling to directly compute the summations and high-dimensional integrals in the Dyson series expansion of the quantum observable \cite{Stein1978b}, and after applying Wick's theorem \cite{Negele1988}, this approach can be represented as the summation of all possible diagrams, each of which is determined by a finite time sequences and a partition of them into pairs. However, such technique may encounter the notorious numerical sign problem \cite{Chen2017,Cai2020,Cai2020b}, meaning that the number of Monte Carlo samples is required to grow at
least exponentially (with respect to physical time) in order to keep the accuracy of the simulation.

Recently, the inchworm Monte Carlo method \cite{Antipov2017,Chen2017,Chen2017b,Cohen2015,Dong2017,Ridley2018} was proposed to mitigate the numerical sign problem. It introduces bold lines as partial resummations of bare dQMC, so that the total number of diagrams can be reduced, and the sign problem is hence suppressed. This approach is further improved in \cite{Cai2020} by writing the evolution of the bold lines as an integro-differential equation, which only requires to sum over ``linked'' diagrams, so that the computational cost can be further reduced. Even after such reductions, however, as the number of points in the time sequence $m$ increases, the total number of diagrams still grows as a double factorial $O((m-1)!!)$. The Monte Carlo sampling of these diagrams will again contribute to the stochastic error, when a large $m$ is needed.

One possible approach to reduce the stochastic error is to sum up all the diagrams with the same time sequence using a deterministic method. As a direct summation is prohibitive due to the large number of diagrams, it calls for designing better algorithms to circumvent the difficulty. The fermionic bath influence functional in the bare continuous-time hybridization expansion (CTHYB) \cite{Muhlbacher2008,Werner2009,Schiro2009,Werner2006}, which is the counterpart of bare dQMC for bosons, can be calculated in the form of a determinant \cite{Muhlbacher2008,Werner2006} and thus the computational cost can be reduced to $O(m^3)$. In the inchworm method, which has less severe numerical sign problem, such a method cannot be directly applied as inchworm expansion only sums over the linked diagrams and thus corresponding bath influence functional cannot be written in a determinant form as in CTHYB.

The recent work~\cite{Boag2018} tackles this challenge with an \emph{inverted algorithm}, which takes the idea of \cite{Rossi2017} that considers the sum of all linked diagrams at once and utilizes the massive cancellations between the diagrams, leading to an exponential rather than factorial computational complexity. It has been shown that the inverted algorithm can asymptotically achieve a computational cost at $O(m^3 \alpha^{m})$, which is significantly smaller than the double factorial complexity for the direct summation of all linked diagrams. The work~\cite{Boag2018} also developed another algorithm based on inclusion-exclusion principle which is even more efficient in the sense that it can further reduce the constant $\alpha$ in the context of inchworm hybridization expansion. The inclusion-exclusion principle describes how the cardinality (or other measures) of unions of sets can be calculated, which is also well known through the Venn diagram. This principle has been applied in a number of areas to reduce the computational complexity, including set partitioning \cite{Bjorklund2009}, counting perfect matchings \cite{Bjorklund2012} and computing matrix permanents \cite{Ryser1963}. The algorithm in \cite{Boag2018} is designed by excluding all the unlinked diagrams from the set of all diagrams, where the set of all the unlinked diagrams is represented by the union of several non-disjoint sets. This allows the inclusion-exclusion principle to be applied to the summation of diagrams, resulting in significant reduction of the computational time.

In this work, we aim to generalize these efforts on fermionic cases to bosonic cases for the simulation of open quantum systems. For bare dQMC, instead of the determinant form, the bath influence functional now holds the form of the matrix hafnian \cite{Barvinok99,Bjorklund2019}. By the inclusion-exclusion principle, we propose a fast algorithm with computational cost $O(2^m)$, which is efficient for small values of $m$ (around $m \leq 20$).
We then further generalize the idea to inchworm method for bosonic systems so that the summation of diagrams with the same time sequence also requires only $O(\alpha^m)$ operations. A sharp estimation of $\alpha$ will also be provided for our algorithm.


The rest of this paper proceeds as follows: In Section \ref{sec:Dyson}, we introduce the Dyson series and use inclusion-exclusion principle to derive an algorithm which sums over the diagrams in the Dyson series efficiently. In Section \ref{sec:inchworm}, another fast algorithm based on inclusion-exclusion principle is designed to sum over the linked diagrams appearing in the inchworm method. An optimization for this algorithm is further proposed, and a complexity analysis is included to  examine the computational cost of the optimized algorithm. Section \ref{sec:num exp} verifies these theoretical results by numerical experiments. Finally, we draw our conclusion in Section \ref{sec:conclusion}.

\section{Dyson series with inclusion-exclusion principle}\label{sec:Dyson}
We study an open quantum system described by the von Neumann equation   
\begin{equation} \label{eq:vonNeumann}
\ii \frac{\dd \rho}{\dd t} = [H, \rho],
\end{equation}
where the density matrix $\rho(t)$ and the Schr\"odinger picture Hamiltonian $H$ above are both Hermitian operators on the Hilbert space $\mc{H} = \mc{H}_s \otimes \mc{H}_b$, with
$\mc{H}_s$ and $\mc{H}_b$ representing respectively the Hilbert spaces associated with the system and the bath of the open quantum system. The Hamiltonian $H$ takes the form as a combination of an uncoupled Hamiltonian $H_0$ and a coupling term $W$. Here we assume that the coupling term $W$ takes the tensor-product form, so that we have
\begin{displaymath}
 H =  H_0 + W  :=  ( H_s \otimes \mathrm{Id}_b + \mathrm{Id}_s \otimes H_b ) +  W_s \otimes W_b,
\end{displaymath}
where $H_s,W_s \in \mc{H}_s$, $H_b,W_b \in \mc{H}_b$, and $\mathrm{Id}_s,\mathrm{Id}_b$ are respectively the
identity operators for the system and the bath.

We are interested in the evolution of the
expectation for a given observable $O = O_s \otimes \mathrm{Id}_b$ acting only on the system part, defined by
\begin{equation} \label{eq:O(t)}
\langle O(t) \rangle := \tr(O \rho(t))
  = \tr(O \ee^{-\ii t H} \rho(0) \ee^{\ii t H}).
\end{equation}
Due to the high dimensionality of the space $\mc{H}_b$, it is usually impractical to solve $\ee^{\pm \ii t H}$ directly. One feasible approach is to apply the method of quantum Monte Carlo to approximate $\langle O(t) \rangle$ numerically. Below we will first introduce the bare dQMC based on the Dyson series expansion of $\langle O(t) \rangle$, and then propose an efficient method to compute a key term in the expansion known as the bath influence functional.

\subsection{Introduction to the bare diagrammatic quantum Monte Carlo method} \label{sec:dQMC}
Upon assuming the initial density matrix has the separable
form $\rho(0) = \rho_s \otimes \rho_b$ where the bath $\rho_b$ commutes with the Hamiltonian $H_b$, the expectation of observable $\langle O(t) \rangle$ can be represented by the following \emph{Dyson series} (for derivation, see \cite{Cai2020}):
\begin{equation} \label{eq:observable1}
   \begin{split} 
& \langle O(t) \rangle = \sum_{m=0}^{+\infty}
  \ii^m \int^{2t}_0 \dd s_m \int^{s_m}_0 \dd s_{m-1} \cdots \int^{s_2}_0  \dd s_1 \ (-1)^{\#\{\sb < t\}} \times \\
  & \hspace{100pt} \times   \tr_s(\rho_s \mathcal{U}^{(0)}(0, s_1,\cdots,s_m, 2t)) \cdot
    \mathcal{L}_b(s_1,\cdots,s_m) .
   \end{split}
\end{equation}
Here $\#\{\sb < t\}$ denotes the number of $s_i$ which are less than $t$ and $\mathrm{tr}_s$ takes trace of the system degree of freedom. 

In practice, one may truncate the series above at a sufficiently large $\bar{M}$ and evaluate those high-dimensional integrals on the right-hand side using Monte Carlo integration, resulting in the bare dQMC. This requires us to evaluate the integrand in the Dyson series \eqref{eq:observable1} for each sample. The explicit formula for the propagator $\mc{U}^{(0)}$ in given in Appendix \ref{app:formulas}, which contains the observable $O_s$ and is associated with the system space. As for the \emph{bath influence functional} $\Ls_b$, we assume that the Wick's theorem can be applied so that
\begin{equation} \label{eq:L all pair}
   \Ls_b(s_1,\cdots,s_m) =  \left\{   \begin{array}{l l}
   0, & \text{if $m$ is odd}; \\ 
   \displaystyle \sum_{\mf{q} \in \mQ(\sb)} \prod_{(s_j,s_k) \in \mf{q}} B(s_j,s_k), &  \text{if $m$ is even},
    \end{array} \right.
\end{equation}
where $B: \{(\tau_1,\tau_2) \mid 0 \leq \tau_1 \leq \tau_2 \} \rightarrow \C$ is the two-point bath correlation and the set $\mQ(s_1,\cdots,s_m)$ is the collection of all possible ordered pairings of the time sequence $(s_1,\cdots,s_m)$:
\begin{equation} \label{eq:all linking pairs}
  \begin{split}
  & \mQ(s_1,\cdots,s_m)
  =\Big\{ \{(s_{j_1}, s_{k_1}), \cdots, (s_{j_{m/2}}, s_{k_{m/2}})\} \,\Big\vert\,  \{j_1, \cdots, j_{m/2}, k_1, \cdots, k_{m/2}\} = \{1,\cdots,m\}, \\
  & \hspace{120pt} j_l < k_l \text{ for any } l = 1,\cdots,m/2
  \Big\}.
  \end{split}
\end{equation}  
For example, when $m=4$, $\Ls_b(s_1,s_2,s_3,s_4)$ is given by
\begin{equation} \label{eq:all pairs example}
   \Ls_b(s_1,s_2,s_3,s_4) = B(s_1,s_2) B(s_3,s_4) + B(s_1,s_3) B(s_2,s_4) + B(s_1,s_4) B(s_2,s_3).
\end{equation}
In particular, when $m = 0$, the value of $\mc{L}_b(\emptyset)$ is defined as $1$. With such expression of bath influence functional, the right-hand side of \eqref{eq:observable1} only sums over the terms with even $m$. We may also express \eqref{eq:all pairs example} using many-body diagrams:
\begin{equation} \label{eq:all linking pair diagram example}
\Ls_b(s_1,s_2,s_3,s_4)
=
\begin{tikzpicture}
\draw[-] (0,0)--(1.5,0);\draw plot[only marks,mark =*, mark options={color=black, scale=0.5}]coordinates {(0,0) (0.5,0) (1,0)(1.5,0)};
\draw[-] (0,0) to[bend left] (0.5,0);
\draw[-] (1,0) to[bend left] (1.5,0);
 \end{tikzpicture}
 +
 \begin{tikzpicture}
\draw[-] (0,0)--(1.5,0);\draw plot[only marks,mark =*, mark options={color=black, scale=0.5}]coordinates {(0,0) (0.5,0) (1,0)(1.5,0)};
\draw[-] (0,0) to[bend left] (1,0);
\draw[-] (0.5,0) to[bend left] (1.5,0);
 \end{tikzpicture}
 +
  \begin{tikzpicture}
\draw[-] (0,0)--(1.5,0);\draw plot[only marks,mark =*, mark options={color=black, scale=0.5}]coordinates {(0,0) (0.5,0) (1,0)(1.5,0)};
\draw[-] (0,0) to[bend left] (1.5,0);
\draw[-] (0.5,0) to[bend left] (1,0);
 \end{tikzpicture} \quad .
\end{equation}
In the diagrammatic representation above, each diagram refers to a product $B(\cdot,\cdot)B(\cdot,\cdot)$ where each arc connecting a pair denotes the corresponding two-point correlation. 

The major challenge on evaluating $\Ls_b(s_1,\cdots,s_m)$ (m is even) is that its diagrammatic representation includes in total $(m-1)!!$ diagrams, which leads to a double factorial growth in the computational cost on calculating such a bath influence functional via \emph{direct method} (i.e., direct summation over each diagram in the expansion  such as \eqref{eq:all linking pair diagram example}). As this cost increases drastically when $m$ gets larger, one needs to compute a given $\Ls_b(s_1,\cdots,s_m)$ using the Monte Carlo method (on top the Monte Carlo sampling of $(s_1, \cdots, s_m)$), leading to larger stochastic error. In this section, we will show how we can benefit from the well-known inclusion-exclusion principle to greatly reduce the complexity of computing the bath influence functional.

\subsection{Inclusion-exclusion principle for computing $\Ls_b(s_1,\cdots,s_m)$}
\label{sec:rec fast algo}
Mathematically, the equation \eqref{eq:L all pair} is known to be the hafnian of an undirected graph \cite{Barvinok99}. Several fast algorithms have been introduced to compute such quantity in the recent years \cite{Bjorklund2008,Kan2008,Koivisto2009,Nederlof2009,Cygan2015,Bjorklund2019}. While most algorithms aim for a good complexity for large values of $m$, here we are going to introduce a novel fast algorithm for computing hafnians with small $m$ based on the inclusion-exclusion principle.

Given a function $\mu(\cdot)$ satisfying the additivity such that $ \mu\left( \bigcup^n_{i=1} E_i \right) = \sum^n_{i=1}\mu(E_i)$ for any disjoint sets $\{E_i\}^{n}_{i=1}$, the inclusion-exclusion principle reads 
\begin{equation}\label{inclusion-exclusion principle}
\mu \left(S\Big\backslash \bigcup^n_{i=1}A_i\right) = \mu(S) - \sum_{i=1}^n \mu(A_i) + \sum_{1\le i < j \le n}\mu(A_i \cap A_j) - \cdots +(-1)^n \mu(A_1 \cap \cdots \cap A_n)
\end{equation} 
where $S$ is a given finite universal set containing $A_1,A_2,\cdots,A_n$.

The inclusion-exclusion principle plays important roles in a number of fast algorithms such as Ryser's algorithm for matrix permanents \cite{Ryser1963} and the diagrammatic resummation of quantum impurity models \cite{Boag2018}. To apply this principle on the evaluation of $\Ls_b(s_1,\cdots,s_m)$, we set $S$ as the collection of all combinations of $m/2$ distinct pairs from $s_1,\cdots,s_m$ and $A_i$ as all combinations of $m/2$ distinct pairs from  $s_1,\cdots,s_m$ except $s_i$:


 \begin{equation*}
    \begin{split}
&S =\Big\{ \left( (x_1,y_1),\cdots,(x_{m/2},y_{m/2}) \right) \ \Big| \   x_j = s_{j_1},y_j = s_{j_2} \in    \{s_1,\cdots,s_{m} \} \text{~with~} j_1 < j_2\\
&\hspace{310pt}\text{for any~} j = 1,\cdots,m/2\Big\} , \\
&A_i =\Big\{ \left( (x_1,y_1),\cdots,(x_{m/2},y_{m/2}) \right) \ \Big| \\
& \hspace{20pt} x_j = s_{j_1},y_j = s_{j_2} \in    \{s_1,\cdots,s_{i-1},s_{i+1},\cdots,s_{m} \} \text{~with~} j_1 < j_2 \text{~for any~} j = 1,\cdots,m/2   \Big\}  .  
\end{split}
\end{equation*}
We point out that
\begin{itemize}
\item For any element of $S$ or $A_i$, one pair may appear multiple times, e.g. $\big((s_1, s_2), (s_1, s_2) \big) \in S$.
\item The elements in $S$ and $A_i$ are ordered: The same pairs arranged in different orders form different elements in these sets, e.g. $\big( (s_1, s_3), (s_1, s_2) \big)$ and $\big( (s_1, s_2), (s_1, s_3) \big)$ are different elements of $S$.
\end{itemize}
Below we provide all these sets for $m=4$ represented by diagrams as an example:
\input{images/pair_example}%
Each diagram in the braces refers to a pairing in $S$ or $A_i$ whose first and second components are represented by \bluetext{blue} and \redtext{red} arcs respectively. For instance, we have 
\begin{displaymath}
         (\bluetext{(s_1,s_2)},\redtext{(s_1,s_2)}) =  \begin{tikzpicture}
\draw[-] (0,0)--(1.5,0);\draw plot[only marks,mark =*, mark options={color=black, scale=0.5}]coordinates {(0,0) (0.5,0) (1,0)(1.5,0)};
\draw[-,blue] (0,0) to[bend left] (0.5,0);
\draw[-,red] (0,0) to[bend left=100] (0.5,0);
    \end{tikzpicture} \ , \ 
  (\bluetext{(s_1,s_2)},\redtext{(s_1,s_3)}) =   \begin{tikzpicture}
\draw[-] (0,0)--(1.5,0);\draw plot[only marks,mark =*, mark options={color=black, scale=0.5}]coordinates {(0,0) (0.5,0) (1,0)(1.5,0)};
\draw[-,blue] (0,0) to[bend left] (0.5,0);
\draw[-,red] (0,0) to[bend left=40] (1,0);
    \end{tikzpicture} \ , \cdots 
\end{displaymath}
One can observe that $s_i$ (the point marked in \greentext{green}) never occurs in $A_i$. Regardless of the \bluetext{blue}/\redtext{red} color of arcs, all diagrams of $A_i$ are excluded from $\mc{Q}(s_1,s_2,s_3,s_4)$, which is the collection of the diagrams in the diagrammatic representation of $\Ls_b(s_1,s_2,s_3,s_4)$. In fact, the union of $A_i$ contains all diagrams that are not in $\mc{Q}(s_1,\cdots,s_m)$, and therefore $S\backslash \bigcup^m_{i=1}A_i$ is the set formed by arranging all the pairings in $\mQ(s_1,\cdots,s_m)$ in all possible orders. Denote a set function $\mu(\cdot)$ as 
\begin{equation*}
    \mu(E):= \sum_{ \left( (x_1,y_1),\cdots,(x_{m/2},y_{m/2}) \right)\in E}  B(x_1,y_1) \cdots   B(x_{m/2},y_{m/2})  \text{~for~} E=S,A_1,\cdots,A_m
\end{equation*} 
the left-hand side of inclusion-exclusion principle \eqref{inclusion-exclusion principle} is then given by 
\begin{equation}\label{eq: all pair inclu-exclu LHS}
    \mu\left(S\Big \backslash \bigcup^m_{i=1}A_i\right) = \Ls_b(s_1,\cdots,s_m)\cdot (m/2)!.
\end{equation}
Here the combinatorial factor  $(m/2)!$ takes into account the  possible permutations of pairs in $((s_{j_1}, s_{k_1}), \cdots, (s_{j_{m/2}}, s_{k_{m/2}}))$ which all refer to the same element in the set. 

On the other hand, the terms in the right-hand side of \eqref{inclusion-exclusion principle} are 
  \begin{equation}\label{eq: all pair inclu-exclu RHS}
    \begin{split}
&\mu(S) =          \Biggl(\sum_{1\le i < j \le m} B(s_i,s_j)   \Biggr)^{m/2}, \\
& \mu(A_{k_1}\cap \cdots \cap A_{k_l}) =\Biggl(     \sum_{\substack{1\le i < j \le m \\ i \neq k_1,j\neq k_1 \\ \cdots \\ i \neq k_l,j\neq k_l} } B(s_i,s_j)     \Biggr)^{m/2} \text{~for~} l = 1,\cdots,m-2 \text{~and~} 1\le k_1 <\cdots < k_l \le m.
    \end{split}
\end{equation}
Note that the intersection $A_{k_1}\cap \cdots \cap A_{k_l}$ does not include $s_{k_1},\cdots,s_{k_l}$ and thus the $i,j$ indices in the subscript of the summation can never be assigned as $k_1,\cdots,k_l$. In addition, $A_{k_1}\cap \cdots \cap A_{k_l}$ is empty for $l=m-1,m$, and therefore the corresponding values of $\mu(\cdot)$ is zero.

At this point, we can combine \eqref{eq: all pair inclu-exclu RHS} with \eqref{eq: all pair inclu-exclu LHS} and reach the following formula:
\begin{theorem}\label{thm:Lb}
Given the increasing time sequence $(s_1,\cdots,s_m)$ with $m$ being an even number, $\Ls_b(s_1,\cdots,s_m)$ defined in \eqref{eq:L all pair} can be calculated by 
\begin{equation}\label{eq: inclu_exclu_v}
   \begin{split}
    &\Ls_b(s_1,\cdots,s_m) = \Bigg[ (Q)^{\frac{m}{2}} - \sum^m_{k_1 = 1}\left( Q_{k_1}\right)^{\frac{m}{2}} + \sum_{1\le k_1 <k_2 \le m} \left( Q_{k_1 k_2}\right)^{\frac{m}{2}} - \cdots\\
    & \hspace{150pt}  \cdots + \sum_{1\le k_1 <\cdots < k_{m-2} \le m}  \left( Q_{k_1 k_2 \cdots k_{m-2}}\right)^{\frac{m}{2}}  \Bigg]\Bigg/ \left(\frac{m}{2}\right)!  
      \end{split}
\end{equation}
where 
\begin{displaymath}
Q =\sum_{1\le i < j \le m} B(s_i,s_j) \quad \text{and}\quad  Q_{k_1 k_2 \cdots k_{n}} =  \sum_{\substack{1\le i < j \le m \\ i \neq k_1,j\neq k_1 \\ \cdots \\ i \neq k_{n},j\neq k_{n}} } B(s_i,s_j) \quad \text{for} \quad n = 1,\cdots,m-2. 
\end{displaymath}
\end{theorem}
The equation \eqref{eq: inclu_exclu_v} can be used to calculate $\Ls_b(s_1,\cdots,s_m)$ for given values of $B(s_i, s_j)$. Again in the example of $m=4$, this equation can be expanded as
\begin{equation}\label{inclu_exclu_v_example}
    \begin{split}
&\Ls_b(s_1,s_2,s_3,s_4) = \frac{1}{2} \Big[ \Big( B(s_1,s_2) + B(s_1,s_3) + B(s_1,s_4) +B(s_2,s_3) + B(s_2,s_4)  + B(s_3,s_4) \Big)^2 \\
& \quad - \Big( B(s_2,s_3) + B(s_2,s_4) + B(s_3,s_4) \Big)^2 - \Big( B(s_1,s_3) + B(s_1,s_4) + B(s_3,s_4) \Big)^2  \\
& \quad - \Big( B(s_1,s_2) + B(s_1,s_4) + B(s_2,s_4) \Big)^2 - \Big( B(s_1,s_2) + B(s_1,s_3) + B(s_2,s_3) \Big)^2  \\
& \quad + \big( B(s_3,s_4) \big)^2 + \big( B(s_2,s_4) \big)^2+ \big( B(s_2,s_3) \big)^2+ \big( B(s_1,s_4) \big)^2+ \big( B(s_1,s_3) \big)^2+ \big( B(s_1,s_2) \big)^2 \Big].
   \end{split}
\end{equation}
It can be checked by direct calculation that the cancellations among these two-point correlations will finally lead to the same result as \eqref{eq:all pairs example} from the definition of $\Ls_b(s_1,\cdots,s_m)$, which justifies the equivalence between the two approaches to evaluate the bath influence functional.

The formula \eqref{inclu_exclu_v_example} does not seem to hold any advantage at first glance. Indeed, for small values of $m$, the formula \eqref{eq: inclu_exclu_v} requires more operations than the direct approach using the definition \eqref{eq:L all pair}. However, the computational cost of such inclusion-exclusion principle will become significantly cheaper when $m$ is large, which is
\begin{multline}\label{inclu exclu v cost original}
    \left( {m \atop 0} \right) \left( {m \atop 2} \right) +  \left( {m \atop 1} \right)  \left( {m-1 \atop 2} \right) + \cdots +  \left( {m \atop m-2} \right)  \left( {2 \atop 2} \right) \\
   =  \sum^{m-2}_{n=0} \frac{m!}{n!(m-n)!}\cdot \frac{(m-n)!}{2(m-n-2)!}  = \frac{m(m-1)}{2}\cdot \sum^{m-2}_{n=0} \left( { m-2 \atop n  }\right)  \sim O(m^2 2^m).
\end{multline}
Here the binomial coefficient $\left( {m \atop n} \right)$ is the number of terms in the summation in \eqref{eq: inclu_exclu_v} with respect to $k_1, \cdots, k_n$, and the binomial coefficient $\left( {{m-n} \atop 2} \right)$ corresponds to the number of terms in the definition of $Q_{k_1 \cdots k_n}$. This cost grows considerably slower than the double factorial for the direct calculation of the bath influence functional. We note that the reduction of computational complexity can be compared to the Ryser formula \cite{Ryser1963} for computing the permanent of an $m\times m$ matrix, which is also derived from the inclusion-exclusion principle with the computational cost also being $O(m^2 2^m)$. 
In our case, we are able to further reduce the computational cost to $O(2^m)$ by calculating $Q_{k_1k_2\cdots k_n}$ iteratively. Specifically, we define the symmetrization of the two-point correlation as 
\begin{align*}
\bB(s_i,s_j)    =\left\{ \begin{array}{l l}
    B(s_i,s_j)     & \text{if~} s_i < s_j,  \\
       0  & \text{if~}  s_i = s_j, \\
       B(s_j,s_i) & \text{if~} s_i > s_j
    \end{array}\right.
\end{align*}
Below we use $\bB \in \C^{m\times m}$ to denote the symmetric matrix whose entries are $\bB(s_i, s_j)$.
Based on such definition, the desired bath influence functional is then computed by the hafnian \cite{Barvinok99} of the symmetric matrix $\bB$, while $Q$ and each $Q_{k_1k_2\cdots k_n}$ are expressed in terms of the summation over entries of $\bB$:
\begin{align*}
Q & = \frac{1}{2} \sum_{i=1}^m\sum_{j=1}^m\bB(s_i,s_j),\\
Q_{k_1k_2\cdots k_n} & =  \frac{1}{2} \sum_{\substack{1\le i \le m\\i\neq k_1\\ \cdots\\i\neq k_n}}\sum_{\substack{1\le j \le m \\j\neq k_1\\ \cdots\\j\neq k_n}}\bB(s_i,s_j).
\end{align*}
Note that the constant $\frac{1}{2}$ is needed here since we have taken into each term twice in the symmetrized summation. From this definition, we can observe that $Q_{k_1 k_2 \cdots k_n}$ is a half of the sum over all the entries of $\bB$ excluding the $k_i$th rows and columns for $i = 1,\cdots,n$. Thus the value of $Q_{k_1k_2\cdots k_n}$ can be obtained from the $Q_{k_1k_2\cdots k_{n-1}}$ by taking away the $k_n$th row and column. We are then inspired to define
\begin{equation} \label{R def}
\begin{split}
R^{(i)}_{k_1k_2\cdots k_{n-1}}  = & \ \sum^m_{j=1} \bB(s_i,s_j) - \underline{\sum^{n-1}_{j=1}\bB(s_i,s_{k_j})}\\
= & \ \frac{1}{2}\left( \sum^m_{j=1} \bB(s_i,s_j) - \underline{\sum^{n-1}_{j=1}\bB(s_i,s_{k_j})} \right) +  \frac{1}{2}\left( \sum^m_{j=1} \bB(s_j,s_i) - \underline{\sum^{n-1}_{j=1}\bB(s_{k_j},s_i)} \right)
\end{split}
\end{equation}
which describes the summation over $i$th row and column except the $k_1$th, $k_2$th, $\cdots$, $k_{n-1}$th entries. The underlined terms are subtracted since they do not exist in the sum $Q_{k_1 k_2 \cdots k_{n-1}}$, and the diagonal entry $\bB(s_i, s_i)$ is counted twice but this does not matter since it equals zero by definition. Now we can compute $Q_{k_1k_2\cdots k_n}$ by 
\begin{equation} \label{iter Q}
  Q_{k_1k_2\cdots k_n} =  Q_{k_1k_2 \cdots k_{n-1}} - R^{(k_n)}_{k_1k_2\cdots k_{n-1}}.
\end{equation}
Note that this relation also holds for $n=1$, for which the left-hand side of \eqref{R def} becomes $R^{(i)}$, denoting the sum of the $i$th row of the matrix $\bB$. The equation \eqref{iter Q} reduces the computational cost of each $Q_{k_1 k_2 \cdots k_n}$ to $O(1)$ once the initial value $Q$ is given, and each $R^{(k_n)}_{k_1k_2\cdots k_{n-1}}$ can also be obtained iteratively by only one subtraction:
\begin{equation}
 \label{iter R}
        R^{(i)}_{k_1k_2\cdots k_n} =   R^{(i)}_{k_1k_2\cdots k_{n-1} }-\bB(s_{k_n},s_i),
\end{equation}
which can be easily seen according to its definition.

For a more intuitive understanding, one may refer to Figure \ref{fig:fast v algo} to visualize the procedures to compute a simple example as $Q_{24}$: Each node is assigned with the value of the corresponding entry of the matrix $\bB$ (with the coefficient $\frac{1}{2}$), so $Q$ is simply the summation over all such nodes. We can reach to the desired $Q_{24}$ by the following two steps:

\begin{itemize}
    \item[\textbf{(i)}] calculate $R^{(2)}$ (summation over the \redtext{red} lines) and obtain $Q_2$ (summation over all nodes that are not on \redtext{red} lines) using relation \eqref{iter Q};
    \item[\textbf{(ii)}] calculate $R^{(4)}_2$ (summation over the \bluetext{blue} lines excluding the two boxed nodes) using relation \eqref{iter R} and get $Q_{24}$ (summation over all nodes that are on neither \redtext{red} nor \bluetext{blue} lines) again by relation \eqref{iter Q}.
\end{itemize}

Note that the nodes on the \greentext{green} diagonal are all equal to zero, which explains why the double counting on the nodes of intersection on \redtext{red}/\bluetext{blue} lines will not affect the result of the calculation as we have mentioned previously.

\begin{figure}[h]
\centering

\begin{tikzpicture}
 
   \draw [thick] (1,1)--(1,6)--(6,6)--(6,1)--(1,1);
   \draw[->,thick] (6,1)--(7,1);
   \draw[->,thick] (1,6)--(1,7);
   \draw[very thick,color=green] (1,1)--(6,6);

   \node[below] at (1,1) {$1$};
   \node[below] at (2,1) {$2$};
   \node[below] at (3,1) {$3$};
   \node[below] at (4,1) {$4$};
   \node[below] at (5,1) {$5$};
   \node[below] at (6,1) {$6$};
   \node[right] at (7,1) {$i$};
   \node[above] at (1,7) {$j$};
 {
    \foreach \i in {1,...,6}
    \foreach \j in {1,...,6}
    \draw plot[only marks,mark=*, mark options={ scale=1}] coordinates {(\i,\j)};
}

\draw[line width=2.5pt,color=red] (1,2)--(6,2);
\draw[line width=2.5pt,color=red] (2,1)--(2,6);

\draw[line width=2.5pt,color=blue] (1,4)--(6,4);
\draw[line width=2.5pt,color=blue] (4,1)--(4,6);    

\draw [thick] (1.8,3.8)--(2.2,3.8)--(2.2,4.2)--(1.8,4.2)--(1.8,3.8);

\draw [thick] (3.8,1.8)--(4.2,1.8)--(4.2,2.2)--(3.8,2.2)--(3.8,1.8);

\node[right] at (6,2) {$\redtext{R^{(2)}}$};
\node[right] at (6,4) {$\bluetext{R^{(4)}_{2}}$};

  \end{tikzpicture}

\caption{An example when $m=6$ illustrating the calculation of $Q_{k_1k_2}$ for $k_1=2$ and $k_2=4$.}
 \label{fig:fast v algo}
\end{figure}
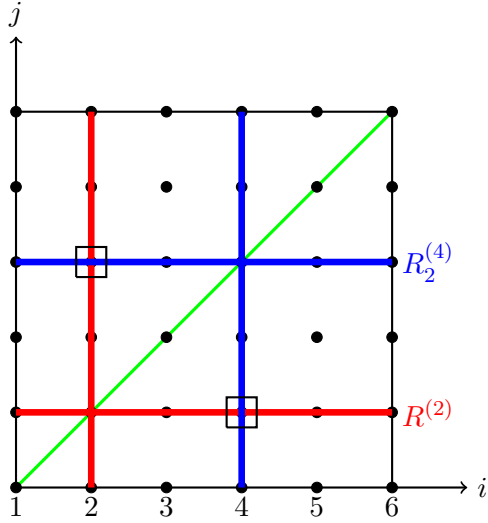
 To end this section, we examine the computational cost of such procedures which are written in details as Algorithm~\ref{algo:rectangular box}: The major complexity concentrates in the calculation of $R^{(i)}_{k_1k_2\cdots k_n}$ and $Q_{k_1k_2\cdots k_n} $ in Line 8 and 9, both of which require 1 subtraction in each iteration for the total $2^m$ iterations, and the evaluation of the final step in Line 13 whose cost is again  $2^m$. Other computations such as the initial settings for $R^{(i)}$ and $Q$ need at most $m^2$ operations and thus are minor. Consequently, Algorithm \ref{algo:rectangular box} has the complexity at $O(2^m)$, which is much cheaper compared to the original \eqref{inclu exclu v cost original}. 

\begin{algorithm}
  \caption{Inclusion-exclusion principle for computing $\Ls_b(s_1,\cdots,s_m)$}\label{algo:rectangular box}
  \begin{algorithmic}[1]
\State Set $Q \gets \frac{1}{2}\sum_{i=1}^m \sum_{j=1}^m\bB(s_i,s_j)$    \Comment{\emph{Initial setting}}
 \medskip
  \For{$i$ from $1$ to $m$}
  \State $R^{(i)} \gets \sum^{m}_{j=1}\bB(s_{i},s_{j})$
  \State $Q_i \gets Q - R^{(i)}$
 \For{$\bar{k}$ from $1$ to $i-1$}
  \For{$n$ from $0$ to $\max( \min(m-4,\bar{k}-1),0)$}
        \For{$1\le k_1 < k_2 < \cdots < k_n < \bar{k}$}
         \medskip
    \State $ R^{(i)}_{k_1\cdots k_n \bar{k}}\gets R^{(i)}_{k_1\cdots k_{n} }-\bB(s_{\bar{k}},s_i)$
    \State $Q_{k_1\cdots k_n \bar{k} i} \gets Q_{k_1 \cdots k_n \bar{k}} - R^{(i)}_{k_1\cdots k_{n} \bar{k}}$
     \medskip
     \medskip
    \EndFor
 \EndFor
 \EndFor
    \EndFor
    \medskip
  \State Compute $\Ls_b(s_1,\cdots,s_m)$ according to \eqref{eq: inclu_exclu_v} \Comment{\emph{Final step}}
  \medskip
  \State \textbf{return} $\Ls_b(s_1,\cdots,s_m)$
   \end{algorithmic}
\end{algorithm}
Compared with previous works on the computation of hafnians, this algorithm does not have the optimal time complexity. In \cite{Bjorklund2019}, Bj\"orklund et al. have proposed an algorithm that computes the hafnian of a $m \times m$ matrix with time complexity $O(m^3 2^{m/2})$, which requires computation of all the eigenvalues of $2^{m/2}$ matrices. Asymptotically, such an algorithm is faster than our algorithm for large $m$. However, according to our experiments, Algorithm \ref{algo:rectangular box} is not slower than Bj\"orklund's algorithm up to $m=22$ due to a relatively smaller prefactor of the overall complexity, which is sufficiently efficient as a satisfactory convergence of Dyson series generally will not require a very large $m$. We have attached our MATLAB code for Algorithm \ref{algo:rectangular box} in Appendix \ref{app:code} and one may compare the time efficiency of Bj\"orklund's algorithm with ours.

\section{Inchworm Monte Carlo method with inclusion-exclusion principle}\label{sec:inchworm}
The bare dQMC method can be only applied to short-time simulation, since the variance of the integrand in the Monte Carlo method grows exponentially with simulation time, which is known as the dynamical sign problem~\cite{Muhlbacher2008,Werner2009,Schiro2010}. One approach to alleviate the sign problem is the \emph{inchworm Monte Carlo method} proposed in \cite{Chen2017}, which introduces the full propagator $G(\Si,\Sf)$ defined by (see \cite{Cai2020} for a derivation)
\begin{equation} \label{eq:DysonG}
   \begin{split}
 &G(\Si, \Sf)  = \sum_{\substack{m=0 \\ m \text{~is even}}}^{+\infty}
  \ii^m \int^{\Sf}_{\Si} \dd s_m \int^{s_m}_{\Si} \dd s_{m-1}\cdots \int^{s_2}_{\Si} \dd s_1 \  \times  \\
  &\hspace{100pt} \times (-1)^{\#\{\sb < t\}}  \mathcal{U}^{(0)}(\Si, s_1,\cdots,s_m, \Sf) \cdot
    \mathcal{L}_b(s_1,\cdots,s_m) 
   \end{split}
\end{equation}
for the initial time point $\Si \in[0,2t]\backslash \{t\}$ and the final time point $\Sf \in [\Si,2t]\backslash \{t\}$. One may compare the definition of such a full propagator with the desired expectation of observable \eqref{eq:observable1} to find the relation $\langle O(t) \rangle = \tr_s (\rho_s G(0,2t))$, suggesting that we should obtain $\langle O(t) \rangle$ by studying the evolution of $G(\Si, \Sf)$.

In \cite[Section 4]{Cai2020}, an integro-differential equation formulation for the full propagator is proposed as 
\begin{multline}\label{eq: inchworm equation}
    \frac{\partial G(\Si,\Sf)}{\partial \Sf} =   \sgn(\Sf - t)  \ii H_s G(\Si,\Sf) +  \sum^{+ \infty}_{\substack{m=2\\ m \text{~is even~}}} \ii^{m} \int^{\Sf}_{\Si} \dd s_{m-1} \int^{s_{m-1}}_{\Si} \dd s_{m-2} \cdots \int^{s_2}_{\Si} \dd s_1 \  \times \\
   \times   (-1)^{\#\{\sb \le t\}}  W_s \mc{U}(\Si,s_1,\cdots,s_{m-1},\Sf)\cdot \Ls_b^c(s_1,\cdots,s_{m-1},\Sf) .
\end{multline} 
Here we recall that $W_s$ is the perturbation associated with the system, and the functional $\mc{U}$ is defined similarly to $\mc{U}^{(0)}$ in \eqref{eq:DysonG} with the bare propagator $G^{(0)}(\cdot, \cdot)$ replaced by the full propagator $G(\cdot,\cdot)$. Its formula together with some important properties of the full propagator are summarized  in Appendix \ref{app:formulas}.

The definition of $\Ls_b^c$ is similar to \eqref{eq:L all pair}:
\begin{equation} \label{eq:L inchworm}
    \Ls_b^c(s_1,\cdots,s_m)   =  \sum_{\mf{q} \in \mQ_c(s_1,\cdots,s_m)} \prod_{(s_j,s_k) \in \mf{q}} B(s_j,s_k),
\end{equation}
where $\mQ_c$ denotes the set of linked pairings:
\begin{equation*}
\mQ_c(s_1,\cdots,s_m) = \{\mf{q} \in \mQ(s_1,\cdots,s_m) \mid \mf{q} \text{ is linked}\}.
\end{equation*}
By saying $\mf{q}$ is ``linked" in the diagrammatic representation, we mean that all points in a diagram are connected with each other using arcs as ``bridges". In the same example \eqref{eq:all linking pair diagram example} as when $m=4$, the second diagram on the right-hand side is considered to be linked since one may start from any of the four points and reach to any other one going through the path formed by the union of the arcs. More rigorously, this linkedness is defined as follows.
\begin{definition}[Linked pairs]
Two pairs of real numbers $(s_1, s_2)$ and $(\tau_1, \tau_2)$ satisfying $s_1
\le s_2$ and $\tau_1 \le \tau_2$ are \emph{linked} if either of the
following two statements holds:
\begin{enumerate}
\item $s_1 \le \tau_1 \le s_2$ and $\tau_1 \le s_2 \le \tau_2$.
\item $\tau_1 \le s_1 \le \tau_2$ and $s_1 \le \tau_2 \le s_2$.
\end{enumerate}
\end{definition}

\begin{definition}[Linked sets of pairs]
Given two sets of pairs $\mf{q}_1$ and $\mf{q}_2$, we say $\mf{q}_1$
and $\mf{q}_2$ are \emph{linked}  if there exists $(s_1, s_2) \in \mf{q}_1$ and $(\tau_1, \tau_2) \in \mf{q}_2$ such that $(s_1, s_2)$ and $(\tau_1, \tau_2)$
are linked. We say a given set of pairs $\mf{q}$ is \emph{linked} if it cannot be decomposed into the union of two sets of pairs that are not linked.
\end{definition}
When two sets of pairs $\mf{q}_1$ and $\mf{q}_2$ are linked, we also say that $\mf{q}_1$ is linked to $\mf{q}_2$ and vice versa. 

\smallskip 

For example, the first diagram on the right-hand side of \eqref{eq:all linking pair diagram example} is not linked since it can be decomposed into $\mf{q}_1 \cup \mf{q}_2$ where $\mf{q}_1 = \{(s_1,s_2)\}$ and $\mf{q}_2 = \{(s_3,s_4)\}$ and obviously $\mf{q}_1$ is not linked to $\mf{q}_2$. For the same reason, the third diagram is not linked either. Therefore, only the second diagram is linked and 
\begin{equation} \label{eq:inchworm example}
           \Ls_b^c(s_1,s_2,s_3,s_4) =B(s_1,s_3) B(s_2,s_4) 
\end{equation}
which is diagrammatically expressed as 
\begin{equation*}
\begin{tikzpicture}
\draw[-] (0,0)--(1.5,0);\draw plot[only marks,mark =*, mark options={color=black, scale=0.5}]coordinates {(0,0) (0.5,0) (1,0)(1.5,0)};
\draw (0,3pt) arc[radius = 3pt,start angle= 90,end angle = 270];
  \draw (1.5,3pt) arc[radius=3pt,start angle=90,end angle=-90];
  \draw[-](0,3pt)--(1.5,3pt);\draw[-](0,-3pt)--(1.5,-3pt);
  \end{tikzpicture}
=
\begin{tikzpicture}
\draw[-] (0,0)--(1.5,0);\draw plot[only marks,mark =*, mark options={color=black, scale=0.5}]coordinates {(0,0) (0.5,0) (1,0)(1.5,0)};
\draw[-] (0,0) to[bend left] (1,0);
\draw[-] (0.5,0) to[bend left] (1.5,0);
 \end{tikzpicture} \; ,
\end{equation*}
where $\Ls_b^c$ is denoted by a rounded box covering time sequence $(s_1,\cdots,s_m)$. Another example for this many-body diagrammatic representation when $m=6$ is given below: 
\begin{equation}\label{linked pairs example}
  \begin{split}
    \Ls_b^c(s_1,s_2,s_3,s_4,s_5,s_6) = & \ \begin{tikzpicture}
\draw[-] (0,0)--(2.5,0);\draw plot[only marks,mark =*, mark options={color=black, scale=0.5}]coordinates {(0,0) (0.5,0) (1,0)(1.5,0)(2,0)(2.5,0)};
\draw (0,3pt) arc[radius = 3pt,start angle= 90,end angle = 270];
  \draw (2.5,3pt) arc[radius=3pt,start angle=90,end angle=-90];
  \draw[-](0,3pt)--(2.5,3pt);\draw[-](0,-3pt)--(2.5,-3pt);
  \end{tikzpicture} \\
  = & \ 
\begin{tikzpicture}
\draw[-] (0,0)--(2.5,0);\draw plot[only marks,mark =*, mark options={color=black, scale=0.5}]coordinates {(0,0) (0.5,0) (1,0)(1.5,0)(2,0)(2.5,0)};
\draw[-] (0,0) to[bend left] (1,0);
\draw[-] (0.5,0) to[bend left] (2,0);
\draw[-] (1.5,0) to[bend left] (2.5,0);
 \end{tikzpicture} 
 +
\begin{tikzpicture}
\draw[-] (0,0)--(2.5,0);\draw plot[only marks,mark =*, mark options={color=black, scale=0.5}]coordinates {(0,0) (0.5,0) (1,0)(1.5,0)(2,0)(2.5,0)};
\draw[-] (0,0) to[bend left] (1.5,0);
\draw[-] (0.5,0) to[bend left] (2,0);
\draw[-] (1,0) to[bend left] (2.5,0);
 \end{tikzpicture}  
 +
 \begin{tikzpicture}
\draw[-] (0,0)--(2.5,0);\draw plot[only marks,mark =*, mark options={color=black, scale=0.5}]coordinates {(0,0) (0.5,0) (1,0)(1.5,0)(2,0)(2.5,0)};
\draw[-] (0,0) to[bend left] (1.5,0);
\draw[-] (0.5,0) to[bend left] (2.5,0);
\draw[-] (1,0) to[bend left] (2,0);
 \end{tikzpicture} 
 +
  \begin{tikzpicture}
\draw[-] (0,0)--(2.5,0);\draw plot[only marks,mark =*, mark options={color=black, scale=0.5}]coordinates {(0,0) (0.5,0) (1,0)(1.5,0)(2,0)(2.5,0)};
\draw[-] (0,0) to[bend left] (2,0);
\draw[-] (0.5,0) to[bend left] (1.5,0);
\draw[-] (1,0) to[bend left] (2.5,0);
 \end{tikzpicture} 
    \end{split}
\end{equation}
We observe that the rounded box above only contains four linked diagrams, while the corresponding bath influence functional $\Ls_b(s_1,s_2,s_3,s_4,s_5,s_6)$ is the summation of all possible ordered pairings including diagrams like 
\begin{equation} \label{unlinked diagrams} 
\begin{tikzpicture}
\draw[-] (0,0)--(2.5,0);\draw plot[only marks,mark =*, mark options={color=black, scale=0.5}]coordinates {(0,0) (0.5,0) (1,0)(1.5,0)(2,0)(2.5,0)};
\draw[-] (0,0) to[bend left] (0.5,0);
\draw[-] (1,0) to[bend left] (2,0);
\draw[-] (1.5,0) to[bend left] (2.5,0);
 \end{tikzpicture} \quad , \quad
 \begin{tikzpicture}
\draw[-] (0,0)--(2.5,0);\draw plot[only marks,mark =*, mark options={color=black, scale=0.5}]coordinates {(0,0) (0.5,0) (1,0)(1.5,0)(2,0)(2.5,0)};
\draw[-] (0,0) to[bend left] (1,0);
\draw[-] (0.5,0) to[bend left] (2.5,0);
\draw[-] (1.5,0) to[bend left] (2,0);
 \end{tikzpicture} \quad, \quad \cdots 
\end{equation}
which are not linked.

Compared with the Dyson series, the advantage of the equation \eqref{eq: inchworm equation} is that its series with respect to $m$ has a faster convergence than that in \eqref{eq:observable1}, leading to a less severe numerical sign problem. Also, the number of diagrams in $\Ls_b^c(s_1,\cdots,s_m)$ grows asymptotically as $\ee^{-1} (m-1)!!$ \cite{Stein1978b}, which is less than the number of diagrams in $\Ls_b(s_1,\cdots,s_m)$. To solve \eqref{eq:L inchworm} numerically, one can apply Runge-Kutta type methods to discretize the time $\Sf$, and the integrals on the right-hand side of \eqref{eq: inchworm equation} are approximated by the Monte Carlo method, especially for large $m$. Thus, it can be expected that most of the computational time is spent on the evaluation of $\Ls_b^c(s_1,\cdots,s_m)$ defined in \eqref{eq:L inchworm}, and hence, a fast algorithm for $\Ls_b^c(s_1,\cdots,s_m)$ is desirable. Below, we are going to combine the result in Section \ref{sec:dQMC} and the technique developed in \cite{Boag2018} to accelerate the computation of $\Ls_b^c$ for large $m$.

\subsection{Inclusion-exclusion principle for computing rounded box $\Ls_b^c(s_1,\cdots,s_m)$}
Since any rounded box covering two points can be directly evaluated by the corresponding two-point correlation, we only consider the calculation of $\Ls_b^c(s_1,\cdots,s_m)$ with even $m \geq 4$ in the rest of of this section. In the inclusion-exclusion principle \eqref{inclusion-exclusion principle}, we set
\begin{equation} \label{inclu exclu set linked}
    \begin{split}
S = & \ \mc{Q}^*(s_1,\cdots,s_m) \\
:= &  \ \Big\{ \{(s_{j_1}, s_{k_1}), \cdots, (s_{j_{m/2}}, s_{k_{m/2}})\} \,\Big\vert\,  \{j_1, \cdots, j_{m/2}, k_1, \cdots, k_{m/2}\} = \{1,\cdots,m\}, \\
  & \hspace{220pt}  k_l - j_l \geq 4 \text{ for any } l = 1,\cdots,m/2
  \Big\}
\end{split}
\end{equation}
equipped with the set function 
\begin{equation*}
    \mu(S) = \sum_{\mf{q} \in S} \prod_{(s_j,s_k) \in \mf{q}} B(s_j,s_k).
\end{equation*}
In \eqref{inclu exclu set linked}, $\mc{Q}^*(s_1,\cdots,s_m)$ is similarly defined as $\mc{Q}(s_1,\cdots,s_m)$ in \eqref{eq:all linking pairs} but does not include any pair formed by two adjacent time points, i.e., we do not consider any pair $(s_{j_l},s_{k_l})$ with $k_l - j_l = 2$ in $S$. This difference in definition is based on the fact that an ordered pairing with an arc connecting two adjacent points such as the ones in \eqref{unlinked diagrams} will never be linked, and thus is not considered in a rounded box. Note that $S = \emptyset$ when $m=2$, which explains why we restrict our discussion for $m\geq 4$ in this section. Upon further introducing
\begin{equation*}
    B^*(s_j,s_k) = \left\{ \begin{array}{l l}
    B(s_j,s_k),   & \text{if $k-j >1$,}  \\
     0 ,  & \text{if $k-j = 0$,}
    \end{array}   \right.
\end{equation*} 
it is obvious that
\begin{equation}\label{def:rectangular box}
   \Ls_b^*(s_1,\cdots,s_m):=  \mu(S)  = \sum_{\mf{q} \in \mQ^*(\sb)} \prod_{(s_j,s_k) \in \mf{q}} B(s_j,s_k) =  \sum_{\mf{q} \in \mQ(\sb)} \prod_{(s_j,s_k) \in \mf{q}} B^*(s_j,s_k) .
\end{equation}
Diagrammatically, we express a given $\Ls_b^*(s_1,\cdots,s_m)$ by a rectangular box. For example when $m=6$, we have 
\begin{equation}
\begin{tikzpicture}
  \draw[-] (0,0)--(1.75,0);
  \draw plot[only marks, mark=*, mark options = {color=black, scale=.5}] coordinates {(0,0)(0.35,0)(0.7,0)(1.05,0)(1.4,0)(1.75,0)};
  \draw[-](-0.1,-0.1)--(-0.1,0.1)--(1.85,0.1)--(1.85,-0.1)--cycle;
\end{tikzpicture} = 
 \begin{tikzpicture}
 \draw[-] (0,0)--(1.75,0);
  \draw plot[only marks, mark=*, mark options = {color=black, scale=.5}] coordinates {(0,0)(0.35,0)(0.7,0)(1.05,0)(1.4,0)(1.75,0)};
\draw[-] (0,0) to[bend left] (0.7,0);
\draw[-] (0.35,0) to[bend left] (1.4,0);
\draw[-] (1.05,0) to[bend left] (1.75,0);
 \end{tikzpicture} \ + \  
  \begin{tikzpicture}
 \draw[-] (0,0)--(1.75,0);
  \draw plot[only marks, mark=*, mark options = {color=black, scale=.5}] coordinates {(0,0)(0.35,0)(0.7,0)(1.05,0)(1.4,0)(1.75,0)};
\draw[-] (0,0) to[bend left] (1.05,0);
\draw[-] (0.35,0) to[bend left] (1.4,0);
\draw[-] (0.7,0) to[bend left] (1.75,0);
 \end{tikzpicture}  \ + \ 
  \begin{tikzpicture}
 \draw[-] (0,0)--(1.75,0);
  \draw plot[only marks, mark=*, mark options = {color=black, scale=.5}] coordinates {(0,0)(0.35,0)(0.7,0)(1.05,0)(1.4,0)(1.75,0)};
\draw[-] (0,0) to[bend left] (1.05,0);
\draw[-] (0.35,0) to[bend left] (1.75,0);
\draw[-] (0.7,0) to[bend left] (1.4,0);
 \end{tikzpicture} \ + \ 
   \begin{tikzpicture}
 \draw[-] (0,0)--(1.75,0);
  \draw plot[only marks, mark=*, mark options = {color=black, scale=.5}] coordinates {(0,0)(0.35,0)(0.7,0)(1.05,0)(1.4,0)(1.75,0)};
\draw[-] (0,0) to[bend left] (1.4,0);
\draw[-] (0.35,0) to[bend left] (1.05,0);
\draw[-] (0.7,0) to[bend left] (1.75,0);
 \end{tikzpicture} \ + \ 
   \begin{tikzpicture}
 \draw[-] (0,0)--(1.75,0);
  \draw plot[only marks, mark=*, mark options = {color=black, scale=.5}] coordinates {(0,0)(0.35,0)(0.7,0)(1.05,0)(1.4,0)(1.75,0)};
\draw[-] (0,0) to[bend left] (1.75,0);
\draw[-] (0.35,0) to[bend left] (1.05,0);
\draw[-] (0.7,0) to[bend left] (1.4,0);
 \end{tikzpicture} \ .
\end{equation}%
Similar as the bath influence functional $\Ls_b(s_1,\cdots,s_m)$, the rectangular box above contains all linked diagrams (the first four diagrams on the right-hand side). However, $\Ls_b^*(s_1,\cdots,s_m)$ has only one unlinked diagram (the last diagram) and does not include any unlinked diagrams with adjacent pairs. Such idea has also been applied in \cite{Boag2018} referred as ``second optimization" to eliminate some unlinked diagrams that will not be used in inchworm method. By writing $\Ls_b^*(s_1,\cdots,s_m)$ as the last formula of \eqref{def:rectangular box}, we can again apply Algorithm \ref{algo:rectangular box} to compute a given $m-$point rectangular box at the complexity of $O(2^m)$ upon setting all entries on the subdiagonal and superdiagonal of the matrix in Figure \ref{fig:fast v algo} to be zero. 

To continue the inclusion-exclusion principle, we further let 
\begin{equation}
    \begin{split}
        A_V = & \  \Big\{\mf{q} \in S \ \Big| \ \mf{q}  \text{~has a linked component } \mf{q}_V \in \mc{Q}^*(V) \Big\}\\
&\hspace{60pt} \text{given } V := (s_{i+1},\cdots,s_{i+2n} )\text{~being a subsequence of~}    (s_1,\cdots,s_{m-1} ),
    \end{split}
\end{equation}
where we have used the short-hand notation $\mc{Q}^*(V)$ to denote $\mc{Q}^*(s_{i+1}, \cdots, s_{i+2n})$. By saying $\mf{q}_V$ is a \emph{linked component} of $\mf{q}$, we mean $\mf{q}_V \subset \mf{q}$ is linked but not linked to $\mf{q} \backslash \mf{q}_V$. Each $A_V$ is a collection of some unlinked diagrams since each of its elements $\mf{q}$ can be decomposed as $\mf{q} = \mf{q}_V \cup (\mf{q}\backslash\mf{q}_V)$ where the two subsets are not linked to each other. Note that it is sufficient to consider point sets $V$ including successive time points $s_{i+1}, \cdots, s_{i+2n}$ since any unlinked diagram contains at least one linked component with only successive time points, and each $V$ should contain at least four points due to the exclusion of the diagrams with arcs connecting adjacent points. For example when $m=10$, we have 
\input{images/measure_example}%
All linked $\mf{q}_V$ (marked in \redtext{red}) in the corresponding $A_V$ are eventually included in the rounded boxes which are calculated as $\Ls_b^c(V)$. The rest of points in $\{s_1,\cdots,s_{10}\}$ are not linked to $\mf{q}_V$ and they build up all possible ordered pairings without any pair of adjacent two points. Therefore, we group these points in the rectangular boxes and compute them by $\Ls_b^*(\{s_1,\cdots,s_{10}\}\backslash V)$. Note that some rectangular boxes may be divided by rounded boxes into several nonadjacent segments and we use ``thin pumps" to connect these segments above the rounded boxes to indicate that the points covered by these segments are in the same $\Ls_b^*$. Consequently, the diagrams on the right-hand side above are expressed as the following formulas: 
\begin{equation}
    \begin{split}
 &  \begin{tikzpicture}
  \draw[-] (0,0)--(3.15,0);
  \draw plot[only marks, mark=*, mark options = {color=black, scale=.5}] coordinates {(0,0)(0.35,0)(0.7,0)(1.05,0)(1.4,0)(1.75,0)(2.1,0)(2.45,0)(2.8,0)(3.15,0)};
  \draw[-](0,0.1)--(1.05,0.1);
  \draw[-](0,-0.1)--(1.05,-0.1);
  \draw (0,0.1) arc[radius = 0.1,start angle= 90,end angle = 270];
  \draw (1.05,0.1) arc[radius = 0.1,start angle= 90,end angle = -90];
  \draw[-](1.3,-0.1)--(1.3,0.1)--(3.25,0.1)--(3.25,-0.1)--cycle;
\end{tikzpicture}
 = \Ls_b^c(s_1,s_2,s_3,s_4) \Ls_b^*(s_5,s_6,s_7,s_8,s_9,s_{10}), \\
 & \begin{tikzpicture}
  \draw[-] (0,0)--(3.15,0);
  \draw plot[only marks, mark=*, mark options = {color=black, scale=.5}] coordinates {(0,0)(0.35,0)(0.7,0)(1.05,0)(1.4,0)(1.75,0)(2.1,0)(2.45,0)(2.8,0)(3.15,0)};
  \draw[-](0.35,0.1)--(1.4,0.1);
  \draw[-](0.35,-0.1)--(1.4,-0.1);
  \draw (0.35,0.1) arc[radius = 0.1,start angle= 90,end angle = 270];
  \draw (1.4,0.1) arc[radius = 0.1,start angle= 90,end angle = -90];
  \draw[-](3.25,-0.1)--(1.65,-0.1)--(1.65,0.15)--(0.1,0.15)--(0.1,-0.1)--(-0.1,-0.1)--(-0.1,0.15)--(-0.1,0.2)--(3.25,0.2)--cycle;
\end{tikzpicture}
 =     \Ls_b^c(s_2,s_3,s_4,s_5) \Ls_b^*(s_1,s_6,s_7,s_8,s_9,s_{10}), \\
&  \begin{tikzpicture}
  \draw[-] (0,0)--(3.15,0);
  \draw plot[only marks, mark=*, mark options = {color=black, scale=.5}] coordinates {(0,0)(0.35,0)(0.7,0)(1.05,0)(1.4,0)(1.75,0)(2.1,0)(2.45,0)(2.8,0)(3.15,0)};
  \draw[-](1.75,0.1)--(2.8,0.1);
  \draw[-](1.75,-0.1)--(2.8,-0.1);
  \draw (1.75,0.1) arc[radius = 0.1,start angle= 90,end angle = 270];
  \draw (2.8,0.1) arc[radius = 0.1,start angle= 90,end angle = -90];
  \draw[-](3.25,-0.1)--(3.05,-0.1)--(3.05,0.15)--(1.5,0.15)--(1.5,-0.1)--(-0.1,-0.1)--(-0.1,0.15)--(-0.1,0.2)--(3.25,0.2)--cycle;
\end{tikzpicture}
 =   \Ls_b^c(s_6,s_7,s_8,s_9) \Ls_b^*(s_1,s_2,s_3,s_4,s_5,s_{10}), \\
& \begin{tikzpicture}
  \draw[-] (0,0)--(3.15,0);
  \draw plot[only marks, mark=*, mark options = {color=black, scale=.5}] coordinates {(0,0)(0.35,0)(0.7,0)(1.05,0)(1.4,0)(1.75,0)(2.1,0)(2.45,0)(2.8,0)(3.15,0)};
  \draw[-](0,0.1)--(1.75,0.1);
  \draw[-](0,-0.1)--(1.75,-0.1);
  \draw (0,0.1) arc[radius = 0.1,start angle= 90,end angle = 270];
  \draw (1.75,0.1) arc[radius = 0.1,start angle= 90,end angle = -90];
  \draw[-](2,-0.1)--(2,0.1)--(3.25,0.1)--(3.25,-0.1)--cycle;
\end{tikzpicture}
 =  \Ls_b^c(s_1,s_2,s_3,s_4,s_5,s_6) \Ls_b^*(s_7,s_8,s_9,s_{10}), \\
 & \begin{tikzpicture}
  \draw[-] (0,0)--(3.15,0);
  \draw plot[only marks, mark=*, mark options = {color=black, scale=.5}] coordinates {(0,0)(0.35,0)(0.7,0)(1.05,0)(1.4,0)(1.75,0)(2.1,0)(2.45,0)(2.8,0)(3.15,0)};
  \draw[-](1.05,0.1)--(2.8,0.1);
  \draw[-](1.05,-0.1)--(2.8,-0.1);
  \draw (1.05,0.1) arc[radius = 0.1,start angle= 90,end angle = 270];
  \draw (2.8,0.1) arc[radius = 0.1,start angle= 90,end angle = -90];
  \draw[-](3.25,-0.1)--(3.05,-0.1)--(3.05,0.15)--(0.8,0.15)--(0.8,-0.1)--(-0.1,-0.1)--(-0.1,0.15)--(-0.1,0.2)--(3.25,0.2)--cycle;
\end{tikzpicture}
 =  \Ls_b^c(s_4,s_5,s_6,s_7,s_8,s_9) \Ls_b^*(s_1,s_2,s_3,s_{10}), \\
&   \begin{tikzpicture}
  \draw[-] (0,0)--(3.15,0);
  \draw plot[only marks, mark=*, mark options = {color=black, scale=.5}] coordinates {(0,0)(0.35,0)(0.7,0)(1.05,0)(1.4,0)(1.75,0)(2.1,0)(2.45,0)(2.8,0)(3.15,0)};
  \draw[-](0.35,0.1)--(2.8,0.1);
  \draw[-](0.35,-0.1)--(2.8,-0.1);
  \draw (0.35,0.1) arc[radius = 0.1,start angle= 90,end angle = 270];
  \draw (2.8,0.1) arc[radius = 0.1,start angle= 90,end angle = -90];
  \draw[-](3.25,-0.1)--(3.05,-0.1)--(3.05,0.15)--(0.1,0.15)--(0.1,-0.1)--(-0.1,-0.1)--(-0.1,0.15)--(-0.1,0.2)--(3.25,0.2)--cycle;
\end{tikzpicture}
 =  \Ls_b^c(s_2,s_3,s_4,s_5,s_6,s_7,s_8,s_9) \Ls_b^*(s_1,s_{10}).
  \end{split}
\end{equation}%
The union of $A_V$ contains all the diagrams in $S$ which are not linked. Note that in \eqref{measure example 8}, we have neglected the diagram with $8-$point rounded box $\mu(A_{\{s_1,\cdots,s_8\}})$, which equals zero as its rectangular box is formed by the two adjacent points and thus $\Ls_b^*(s_9,s_{10})=0$. Moreover, in the definition of $A_V$ in \eqref{inclu exclu set linked}, we do not need to consider the case where $V$ includes the last time point $s_m$. This is because, given any $V' = \{s_{m-2n+1},\cdots,s_{m}\}$, if we define $A_{V'}$ as the set of diagrams with a linked component including all the last $2n$ points, then each element in $A_{V'}$ can be found in at least one of the $A_V$'s defined in \eqref{inclu exclu set linked}. For example, the five diagrams of $ A_{\{ s_7,s_8,s_9,s_{10}\}}$ given by
\begin{gather*}
\begin{tikzpicture}
  \draw[-] (0,0)--(3.15,0);
  \draw plot[only marks, mark=*, mark options = {color=black, scale=.5}] coordinates {(0,0)(0.35,0)(0.7,0)(1.05,0)(1.4,0)(1.75,0)(2.1,0)(2.45,0)(2.8,0)(3.15,0)};
  \draw[-](0,0)to [bend left](0.7,0);
  \draw[-](0.35,0)to [bend left](1.4,0);
  \draw[-](1.05,0)to [bend left](1.75,0);
  \draw[-,red](2.1,0)to [bend left](2.8,0);
  \draw[-,red](2.45,0)to [bend left](3.15,0);
\end{tikzpicture}
\ ,  \ 
 \begin{tikzpicture}
  \draw[-] (0,0)--(3.15,0);
  \draw plot[only marks, mark=*, mark options = {color=black, scale=.5}] coordinates {(0,0)(0.35,0)(0.7,0)(1.05,0)(1.4,0)(1.75,0)(2.1,0)(2.45,0)(2.8,0)(3.15,0)};
  \draw[-](0,0)to [bend left](1.05,0);
  \draw[-](0.35,0)to [bend left](1.4,0);
  \draw[-](0.7,0)to [bend left](1.75,0);
  \draw[-,red](2.1,0)to [bend left](2.8,0);
  \draw[-,red](2.45,0)to [bend left](3.15,0);
\end{tikzpicture}
 \ , \ 
\begin{tikzpicture}
  \draw[-] (0,0)--(3.15,0);
  \draw plot[only marks, mark=*, mark options = {color=black, scale=.5}] coordinates {(0,0)(0.35,0)(0.7,0)(1.05,0)(1.4,0)(1.75,0)(2.1,0)(2.45,0)(2.8,0)(3.15,0)};
  \draw[-](0,0)to [bend left](1.05,0);
  \draw[-](0.35,0)to [bend left](1.75,0);
  \draw[-](0.7,0)to [bend left](1.4,0);
  \draw[-,red](2.1,0)to [bend left](2.8,0);
  \draw[-,red](2.45,0)to [bend left](3.15,0);
\end{tikzpicture}\\
 \begin{tikzpicture}
  \draw[-] (0,0)--(3.15,0);
  \draw plot[only marks, mark=*, mark options = {color=black, scale=.5}] coordinates {(0,0)(0.35,0)(0.7,0)(1.05,0)(1.4,0)(1.75,0)(2.1,0)(2.45,0)(2.8,0)(3.15,0)};
  \draw[-](0,0)to [bend left](1.4,0);
  \draw[-](0.35,0)to [bend left](1.05,0);
  \draw[-](0.7,0)to [bend left](1.75,0);
  \draw[-,red](2.1,0)to [bend left](2.8,0);
  \draw[-,red](2.45,0)to [bend left](3.15,0);
\end{tikzpicture}
\ , \ 
\begin{tikzpicture}
  \draw[-] (0,0)--(3.15,0);
  \draw plot[only marks, mark=*, mark options = {color=black, scale=.5}] coordinates {(0,0)(0.35,0)(0.7,0)(1.05,0)(1.4,0)(1.75,0)(2.1,0)(2.45,0)(2.8,0)(3.15,0)};
  \draw[-](0,0)to [bend left](1.75,0);
  \draw[-](0.35,0)to [bend left](1.05,0);
  \draw[-](0.7,0)to [bend left](1.4,0);
  \draw[-,red](2.1,0)to [bend left](2.8,0);
  \draw[-,red](2.45,0)to [bend left](3.15,0);
\end{tikzpicture}
\end{gather*}
can be found in $A_{\{s_1,\cdots,s_6\}}$ and $A_{\{s_2,s_3,s_4,s_5\}}$; the set $A_{\{s_5,\cdots,s_{10}\}}$ has four elements which are all included in $A_{\{s_1,s_2,s_3,s_4\}}$. As a result, on the left-hand side of the inclusion-exclusion principle \eqref{inclusion-exclusion principle} we have 
\begin{equation}\label{eq: inchworm inclu-exclu LHS}
    \mu\left( S\Big \backslash \bigcup_{V} A_V \right)= \mu\left(\mc{Q}_c(s_1,\cdots,s_m) \right) = \Ls_b^c(s_1,\cdots,s_m).
\end{equation}
On the right-hand side of \eqref{inclusion-exclusion principle}, if the sets $V_1, \cdots, V_l$ are mutually disjoint, we have
\begin{equation}\label{eq: inchworm inclu-exclu RHS}
  A_{V_1} \cap \cdots \cap A_{V_l} = \left\{\mf{q} \in S \ \Big| \ \mf{q} \text{ has $l$ linked components } \mf{q}_{V_i} \in \mc{Q}^*(V_i) \text{ for } i = 1,\cdots,l \right\};
\end{equation}
if any $V_i$ and $V_j$ contain a common point $s_i$, then $A_{V_1} \cap \cdots \cap A_{V_l} = \emptyset$, which has no contribution in the inclusion-exclusion principle.
Thus, in the example of $m=10$, the following intersections provide nonzero contribution:
\begin{equation}\label{measure example intersection}
    \begin{split}
\mu\left( A_{\{s_2,s_3,s_4,s_5\}} \cap A_{\{s_6,s_7,s_8,s_9\}} \right)=  & \  
\begin{tikzpicture}
  \draw[-] (0,0)--(3.15,0);
  \draw plot[only marks, mark=*, mark options = {color=black, scale=.5}] coordinates {(0,0)(0.35,0)(0.7,0)(1.05,0)(1.4,0)(1.75,0)(2.1,0)(2.45,0)(2.8,0)(3.15,0)};
  \draw[-](0,0)to [bend left](3.15,0);
  \draw[-,red](0.35,0)to [bend left](1.05,0);
  \draw[-,red](0.7,0)to [bend left](1.4,0);
  \draw[-,red](1.75,0)to [bend left](2.45,0);
  \draw[-,red](2.1,0)to [bend left](2.8,0);
\end{tikzpicture}
=  \begin{tikzpicture}
  \draw[-] (0,0)--(3.15,0);
  \draw plot[only marks, mark=*, mark options = {color=black, scale=.5}] coordinates {(0,0)(0.35,0)(0.7,0)(1.05,0)(1.4,0)(1.75,0)(2.1,0)(2.45,0)(2.8,0)(3.15,0)};
  \draw[-](0.35,0.1)--(1.4,0.1);
  \draw[-](0.35,-0.1)--(1.4,-0.1);
  \draw (0.35,0.1) arc[radius = 0.1,start angle= 90,end angle = 270];
  \draw (1.4,0.1) arc[radius = 0.1,start angle= 90,end angle = -90];
  \draw[-](1.75,0.1)--(2.8,0.1);
  \draw[-](1.75,-0.1)--(2.8,-0.1);
  \draw (1.75,0.1) arc[radius = 0.1,start angle= 90,end angle = 270];
  \draw (2.8,0.1) arc[radius = 0.1,start angle= 90,end angle = -90];
  \draw[-](3.25,-0.1)--(3.05,-0.1)--(3.05,0.15)--(0.1,0.15)--(0.1,-0.1)--(-0.1,-0.1)--(-0.1,0.15)--(-0.1,0.2)--(3.25,0.2)--cycle;
\end{tikzpicture} \\
 = & \ \Ls_b^c(s_2,s_3,s_4,s_5)\Ls_b^c(s_6,s_7,s_8,s_9)\Ls_b^*(s_1,s_{10}) ,\\
 \mu\left( A_{\{s_1,s_2,s_3,s_4\}} \cap A_{\{s_6,s_7,s_8,s_9\}} \right)=  & \  
\begin{tikzpicture}
  \draw[-] (0,0)--(3.15,0);
  \draw plot[only marks, mark=*, mark options = {color=black, scale=.5}] coordinates {(0,0)(0.35,0)(0.7,0)(1.05,0)(1.4,0)(1.75,0)(2.1,0)(2.45,0)(2.8,0)(3.15,0)};
  \draw[-](1.4,0)to [bend left](3.15,0);
  \draw[-,red](0,0)to [bend left](0.7,0);
  \draw[-,red](0.35,0)to [bend left](1.05,0);
  \draw[-,red](1.75,0)to [bend left](2.45,0);
  \draw[-,red](2.1,0)to [bend left](2.8,0);
\end{tikzpicture}
=  \begin{tikzpicture}
  \draw[-] (0,0)--(3.15,0);
  \draw plot[only marks, mark=*, mark options = {color=black, scale=.5}] coordinates {(0,0)(0.35,0)(0.7,0)(1.05,0)(1.4,0)(1.75,0)(2.1,0)(2.45,0)(2.8,0)(3.15,0)};
  \draw[-](0,0.1)--(1.05,0.1);
  \draw[-](0,-0.1)--(1.05,-0.1);
  \draw (0,0.1) arc[radius = 0.1,start angle= 90,end angle = 270];
  \draw (1.05,0.1) arc[radius = 0.1,start angle= 90,end angle = -90];
  \draw[-](1.75,0.1)--(2.8,0.1);
  \draw[-](1.75,-0.1)--(2.8,-0.1);
  \draw (1.75,0.1) arc[radius = 0.1,start angle= 90,end angle = 270];
  \draw (2.8,0.1) arc[radius = 0.1,start angle= 90,end angle = -90];
  \draw[-](3.25,-0.1)--(3.05,-0.1)--(3.05,0.15)--(1.5,0.15)--(1.5,-0.1)--(1.3,-0.1)--(1.3,0.15)--(1.3,0.2)--(3.25,0.2)--cycle;
\end{tikzpicture} \\
 = & \ \Ls_b^c(s_1,s_2,s_3,s_4)\Ls_b^c(s_6,s_7,s_8,s_9)\Ls_b^*(s_5,s_{10}) .
  \end{split}
\end{equation}%
We have again neglected the intersection $\mu\left( A_{\{s_1,s_2,s_3,s_4\}} \cap A_{\{s_5,s_6,s_7,s_8\}} \right)$ which contains the null-valued rectangular box $\Ls_b^*(s_9,s_{10})$. Now we insert \eqref{measure example 4}--\eqref{measure example 8} and \eqref{measure example intersection} into \eqref{inclusion-exclusion principle} to get the diagrammatic representation for $\Ls_b^c(s_1,\cdots,s_{10})$ using inclusion-exclusion principle as 
\begin{multline}\label{diagram example}
\begin{tikzpicture}
  \draw[-] (0,0)--(3.15,0);
  \draw plot[only marks, mark=*, mark options = {color=black, scale=.5}] coordinates {(0,0)(0.35,0)(0.7,0)(1.05,0)(1.4,0)(1.75,0)(2.1,0)(2.45,0)(2.8,0)(3.15,0)};
  \draw[-](0,0.1)--(3.15,0.1);
  \draw[-](0,-0.1)--(3.15,-0.1);
  \draw (0,0.1) arc[radius = 0.1,start angle= 90,end angle = 270];
  \draw (3.15,0.1) arc[radius = 0.1,start angle= 90,end angle = -90];
\end{tikzpicture}
=
\begin{tikzpicture}
  \draw[-] (0,0)--(3.15,0);
  \draw plot[only marks, mark=*, mark options = {color=black, scale=.5}] coordinates {(0,0)(0.35,0)(0.7,0)(1.05,0)(1.4,0)(1.75,0)(2.1,0)(2.45,0)(2.8,0)(3.15,0)};
  \draw[-](-0.1,-0.1)--(-0.1,0.1)--(3.25,0.1)--(3.25,-0.1)--cycle;
\end{tikzpicture}
-
\begin{tikzpicture}
  \draw[-] (0,0)--(3.15,0);
  \draw plot[only marks, mark=*, mark options = {color=black, scale=.5}] coordinates {(0,0)(0.35,0)(0.7,0)(1.05,0)(1.4,0)(1.75,0)(2.1,0)(2.45,0)(2.8,0)(3.15,0)};
  \draw[-](0,0.1)--(1.05,0.1);
  \draw[-](0,-0.1)--(1.05,-0.1);
  \draw (0,0.1) arc[radius = 0.1,start angle= 90,end angle = 270];
  \draw (1.05,0.1) arc[radius = 0.1,start angle= 90,end angle = -90];
  \draw[-](1.3,-0.1)--(1.3,0.1)--(3.25,0.1)--(3.25,-0.1)--cycle;
\end{tikzpicture}
-
\begin{tikzpicture}
  \draw[-] (0,0)--(3.15,0);
  \draw plot[only marks, mark=*, mark options = {color=black, scale=.5}] coordinates {(0,0)(0.35,0)(0.7,0)(1.05,0)(1.4,0)(1.75,0)(2.1,0)(2.45,0)(2.8,0)(3.15,0)};
  \draw[-](0.35,0.1)--(1.4,0.1);
  \draw[-](0.35,-0.1)--(1.4,-0.1);
  \draw (0.35,0.1) arc[radius = 0.1,start angle= 90,end angle = 270];
  \draw (1.4,0.1) arc[radius = 0.1,start angle= 90,end angle = -90];
  \draw[-](3.25,-0.1)--(1.65,-0.1)--(1.65,0.15)--(0.1,0.15)--(0.1,-0.1)--(-0.1,-0.1)--(-0.1,0.15)--(-0.1,0.2)--(3.25,0.2)--cycle;
\end{tikzpicture} \\
 \quad -
\begin{tikzpicture}
  \draw[-] (0,0)--(3.15,0);
  \draw plot[only marks, mark=*, mark options = {color=black, scale=.5}] coordinates {(0,0)(0.35,0)(0.7,0)(1.05,0)(1.4,0)(1.75,0)(2.1,0)(2.45,0)(2.8,0)(3.15,0)};
  \draw[-](0.7,0.1)--(1.75,0.1);
  \draw[-](0.7,-0.1)--(1.75,-0.1);
  \draw (0.7,0.1) arc[radius = 0.1,start angle= 90,end angle = 270];
  \draw (1.75,0.1) arc[radius = 0.1,start angle= 90,end angle = -90];
  \draw[-](3.25,-0.1)--(2,-0.1)--(2,0.15)--(0.45,0.15)--(0.45,-0.1)--(-0.1,-0.1)--(-0.1,0.15)--(-0.1,0.2)--(3.25,0.2)--cycle;
\end{tikzpicture}
-
\begin{tikzpicture}
  \draw[-] (0,0)--(3.15,0);
  \draw plot[only marks, mark=*, mark options = {color=black, scale=.5}] coordinates {(0,0)(0.35,0)(0.7,0)(1.05,0)(1.4,0)(1.75,0)(2.1,0)(2.45,0)(2.8,0)(3.15,0)};
  \draw[-](1.05,0.1)--(2.1,0.1);
  \draw[-](1.05,-0.1)--(2.1,-0.1);
  \draw (1.05,0.1) arc[radius = 0.1,start angle= 90,end angle = 270];
  \draw (2.1,0.1) arc[radius = 0.1,start angle= 90,end angle = -90];
  \draw[-](3.25,-0.1)--(2.35,-0.1)--(2.35,0.15)--(0.8,0.15)--(0.8,-0.1)--(-0.1,-0.1)--(-0.1,0.15)--(-0.1,0.2)--(3.25,0.2)--cycle;
\end{tikzpicture}
-
\begin{tikzpicture}
  \draw[-] (0,0)--(3.15,0);
  \draw plot[only marks, mark=*, mark options = {color=black, scale=.5}] coordinates {(0,0)(0.35,0)(0.7,0)(1.05,0)(1.4,0)(1.75,0)(2.1,0)(2.45,0)(2.8,0)(3.15,0)};
  \draw[-](1.4,0.1)--(2.45,0.1);
  \draw[-](1.4,-0.1)--(2.45,-0.1);
  \draw (1.4,0.1) arc[radius = 0.1,start angle= 90,end angle = 270];
  \draw (2.45,0.1) arc[radius = 0.1,start angle= 90,end angle = -90];
  \draw[-](3.25,-0.1)--(2.7,-0.1)--(2.7,0.15)--(1.15,0.15)--(1.15,-0.1)--(-0.1,-0.1)--(-0.1,0.15)--(-0.1,0.2)--(3.25,0.2)--cycle;
\end{tikzpicture}
-
\begin{tikzpicture}
  \draw[-] (0,0)--(3.15,0);
  \draw plot[only marks, mark=*, mark options = {color=black, scale=.5}] coordinates {(0,0)(0.35,0)(0.7,0)(1.05,0)(1.4,0)(1.75,0)(2.1,0)(2.45,0)(2.8,0)(3.15,0)};
  \draw[-](1.75,0.1)--(2.8,0.1);
  \draw[-](1.75,-0.1)--(2.8,-0.1);
  \draw (1.75,0.1) arc[radius = 0.1,start angle= 90,end angle = 270];
  \draw (2.8,0.1) arc[radius = 0.1,start angle= 90,end angle = -90];
  \draw[-](3.25,-0.1)--(3.05,-0.1)--(3.05,0.15)--(1.5,0.15)--(1.5,-0.1)--(-0.1,-0.1)--(-0.1,0.15)--(-0.1,0.2)--(3.25,0.2)--cycle;
\end{tikzpicture} \\
 \quad -
\begin{tikzpicture}
  \draw[-] (0,0)--(3.15,0);
  \draw plot[only marks, mark=*, mark options = {color=black, scale=.5}] coordinates {(0,0)(0.35,0)(0.7,0)(1.05,0)(1.4,0)(1.75,0)(2.1,0)(2.45,0)(2.8,0)(3.15,0)};
  \draw[-](0,0.1)--(1.75,0.1);
  \draw[-](0,-0.1)--(1.75,-0.1);
  \draw (0,0.1) arc[radius = 0.1,start angle= 90,end angle = 270];
  \draw (1.75,0.1) arc[radius = 0.1,start angle= 90,end angle = -90];
  \draw[-](2,-0.1)--(2,0.1)--(3.25,0.1)--(3.25,-0.1)--cycle;
\end{tikzpicture}
-
\begin{tikzpicture}
  \draw[-] (0,0)--(3.15,0);
  \draw plot[only marks, mark=*, mark options = {color=black, scale=.5}] coordinates {(0,0)(0.35,0)(0.7,0)(1.05,0)(1.4,0)(1.75,0)(2.1,0)(2.45,0)(2.8,0)(3.15,0)};
  \draw[-](0.35,0.1)--(2.1,0.1);
  \draw[-](0.35,-0.1)--(2.1,-0.1);
  \draw (0.35,0.1) arc[radius = 0.1,start angle= 90,end angle = 270];
  \draw (2.1,0.1) arc[radius = 0.1,start angle= 90,end angle = -90];
  \draw[-](3.25,-0.1)--(2.35,-0.1)--(2.35,0.15)--(0.1,0.15)--(0.1,-0.1)--(-0.1,-0.1)--(-0.1,0.15)--(-0.1,0.2)--(3.25,0.2)--cycle;
\end{tikzpicture}
-
\begin{tikzpicture}
  \draw[-] (0,0)--(3.15,0);
  \draw plot[only marks, mark=*, mark options = {color=black, scale=.5}] coordinates {(0,0)(0.35,0)(0.7,0)(1.05,0)(1.4,0)(1.75,0)(2.1,0)(2.45,0)(2.8,0)(3.15,0)};
  \draw[-](0.7,0.1)--(2.45,0.1);
  \draw[-](0.7,-0.1)--(2.45,-0.1);
  \draw (0.7,0.1) arc[radius = 0.1,start angle= 90,end angle = 270];
  \draw (2.45,0.1) arc[radius = 0.1,start angle= 90,end angle = -90];
  \draw[-](3.25,-0.1)--(2.7,-0.1)--(2.7,0.15)--(0.45,0.15)--(0.45,-0.1)--(-0.1,-0.1)--(-0.1,0.15)--(-0.1,0.2)--(3.25,0.2)--cycle;
\end{tikzpicture}
-
\begin{tikzpicture}
  \draw[-] (0,0)--(3.15,0);
  \draw plot[only marks, mark=*, mark options = {color=black, scale=.5}] coordinates {(0,0)(0.35,0)(0.7,0)(1.05,0)(1.4,0)(1.75,0)(2.1,0)(2.45,0)(2.8,0)(3.15,0)};
  \draw[-](1.05,0.1)--(2.8,0.1);
  \draw[-](1.05,-0.1)--(2.8,-0.1);
  \draw (1.05,0.1) arc[radius = 0.1,start angle= 90,end angle = 270];
  \draw (2.8,0.1) arc[radius = 0.1,start angle= 90,end angle = -90];
  \draw[-](3.25,-0.1)--(3.05,-0.1)--(3.05,0.15)--(0.8,0.15)--(0.8,-0.1)--(-0.1,-0.1)--(-0.1,0.15)--(-0.1,0.2)--(3.25,0.2)--cycle;
\end{tikzpicture} \\
 \quad -
\begin{tikzpicture}
  \draw[-] (0,0)--(3.15,0);
  \draw plot[only marks, mark=*, mark options = {color=black, scale=.5}] coordinates {(0,0)(0.35,0)(0.7,0)(1.05,0)(1.4,0)(1.75,0)(2.1,0)(2.45,0)(2.8,0)(3.15,0)};
  \draw[-](0.35,0.1)--(2.8,0.1);
  \draw[-](0.35,-0.1)--(2.8,-0.1);
  \draw (0.35,0.1) arc[radius = 0.1,start angle= 90,end angle = 270];
  \draw (2.8,0.1) arc[radius = 0.1,start angle= 90,end angle = -90];
  \draw[-](3.25,-0.1)--(3.05,-0.1)--(3.05,0.15)--(0.1,0.15)--(0.1,-0.1)--(-0.1,-0.1)--(-0.1,0.15)--(-0.1,0.2)--(3.25,0.2)--cycle;
\end{tikzpicture}
+
\begin{tikzpicture}
  \draw[-] (0,0)--(3.15,0);
  \draw plot[only marks, mark=*, mark options = {color=black, scale=.5}] coordinates {(0,0)(0.35,0)(0.7,0)(1.05,0)(1.4,0)(1.75,0)(2.1,0)(2.45,0)(2.8,0)(3.15,0)};
  \draw[-](0.35,0.1)--(1.4,0.1);
  \draw[-](0.35,-0.1)--(1.4,-0.1);
  \draw (0.35,0.1) arc[radius = 0.1,start angle= 90,end angle = 270];
  \draw (1.4,0.1) arc[radius = 0.1,start angle= 90,end angle = -90];
  \draw[-](1.75,0.1)--(2.8,0.1);
  \draw[-](1.75,-0.1)--(2.8,-0.1);
  \draw (1.75,0.1) arc[radius = 0.1,start angle= 90,end angle = 270];
  \draw (2.8,0.1) arc[radius = 0.1,start angle= 90,end angle = -90];
  \draw[-](3.25,-0.1)--(3.05,-0.1)--(3.05,0.15)--(0.1,0.15)--(0.1,-0.1)--(-0.1,-0.1)--(-0.1,0.15)--(-0.1,0.2)--(3.25,0.2)--cycle;
\end{tikzpicture}
 +  
\begin{tikzpicture}
  \draw[-] (0,0)--(3.15,0);
  \draw plot[only marks, mark=*, mark options = {color=black, scale=.5}] coordinates {(0,0)(0.35,0)(0.7,0)(1.05,0)(1.4,0)(1.75,0)(2.1,0)(2.45,0)(2.8,0)(3.15,0)};
  \draw[-](0,0.1)--(1.05,0.1);
  \draw[-](0,-0.1)--(1.05,-0.1);
  \draw (0,0.1) arc[radius = 0.1,start angle= 90,end angle = 270];
  \draw (1.05,0.1) arc[radius = 0.1,start angle= 90,end angle = -90];
  \draw[-](1.75,0.1)--(2.8,0.1);
  \draw[-](1.75,-0.1)--(2.8,-0.1);
  \draw (1.75,0.1) arc[radius = 0.1,start angle= 90,end angle = 270];
  \draw (2.8,0.1) arc[radius = 0.1,start angle= 90,end angle = -90];
  \draw[-](3.25,-0.1)--(3.05,-0.1)--(3.05,0.15)--(1.5,0.15)--(1.5,-0.1)--(1.3,-0.1)--(1.3,0.15)--(1.3,0.2)--(3.25,0.2)--cycle;
\end{tikzpicture}~.
\end{multline}


At this moment, we have obtained an indirect approach based on inclusion-exclusion principle to evaluate a given rounded box $\Ls_b^c(s_1,\cdots,s_m)$: One may first expand a rounded box as \eqref{diagram example} and then compute the (bridged) rectangular boxes using formula \eqref{eq: inclu_exclu_v} (or more efficiently Algorithm \ref{algo:rectangular box}). Each rounded box on the right-hand side with length greater than $2$ can again be evaluated by the same procedure. In the subsequent section, we will further optimize the computational complexity of the calculation by an improved algorithm.

\subsection{Improved algorithm}
The main idea of our improved algorithm is to combine the diagrams with the same rectangular boxes. Specifically, in the example \eqref{diagram example}, the first two diagrams in the last line have the same rectangular box, and they can be written together by the distributive law as $[\mc{L}_b^c(s_2, s_3, s_4, s_5) \mc{L}_b^c(s_6,s_7,s_8,s_9) - \mc{L}_b^c(s_2, \cdots, s_9)] \mc{L}^*_b(s_1, s_{10})$. For simplicity of notations, we define the dotted box:
\begin{equation}\label{dotted example 1}
  \begin{split}
  \begin{tikzpicture}
  \draw[-] (0,0)--(3.15,0);
  \draw plot[only marks, mark=*, mark options = {color=black, scale=.5}] coordinates {(0,0)(0.35,0)(0.7,0)(1.05,0)(1.4,0)(1.75,0)(2.1,0)(2.45,0)(2.8,0)(3.15,0)};
  \draw[dotted,thick](0.25,-0.1)--(0.25,0.1)--(2.9,0.1)--(2.9,-0.1)--cycle;
  \end{tikzpicture} :={} &
  \begin{tikzpicture}
  \draw[-] (0,0)--(3.15,0);
  \draw plot[only marks, mark=*, mark options = {color=black, scale=.5}] coordinates {(0,0)(0.35,0)(0.7,0)(1.05,0)(1.4,0)(1.75,0)(2.1,0)(2.45,0)(2.8,0)(3.15,0)};
  \draw[-](0.35,0.1)--(1.4,0.1);
  \draw[-](0.35,-0.1)--(1.4,-0.1);
  \draw (0.35,0.1) arc[radius = 0.1,start angle= 90,end angle = 270];
  \draw (1.4,0.1) arc[radius = 0.1,start angle= 90,end angle = -90];
  \draw[-](1.75,0.1)--(2.8,0.1);
  \draw[-](1.75,-0.1)--(2.8,-0.1);
  \draw (1.75,0.1) arc[radius = 0.1,start angle= 90,end angle = 270];
  \draw (2.8,0.1) arc[radius = 0.1,start angle= 90,end angle = -90];
  \end{tikzpicture} -
  \begin{tikzpicture}
  \draw[-] (0,0)--(3.15,0);
  \draw plot[only marks, mark=*, mark options = {color=black, scale=.5}] coordinates {(0,0)(0.35,0)(0.7,0)(1.05,0)(1.4,0)(1.75,0)(2.1,0)(2.45,0)(2.8,0)(3.15,0)};
  \draw[-](0.35,0.1)--(2.8,0.1);
  \draw[-](0.35,-0.1)--(2.8,-0.1);
  \draw (0.35,0.1) arc[radius = 0.1,start angle= 90,end angle = 270];
  \draw (2.8,0.1) arc[radius = 0.1,start angle= 90,end angle = -90];
  \end{tikzpicture} \\
  ={} & \mc{L}_b^c(s_2, s_3, s_4, s_5) \mc{L}_b^c(s_6, s_7, s_8, s_9) - \mc{L}_b^c(s_2, s_3, s_4, s_5, s_6, s_7, s_8, s_9).
      \end{split}
\end{equation}%
In general, a dotted box with even time points represents the sum of all partitions of these time points by rounded boxes with length greater than or equal to $4$, and the sign of each term depends on the number of rounded boxes (even for $+$ and odd for $-$). Two more examples are given below:
\begin{equation*} 
\begin{aligned}
& \begin{tikzpicture}
  \draw[-] (0,0)--(3.15,0);
  \draw plot[only marks, mark=*, mark options = {color=black, scale=.5}] coordinates {(0,0)(0.35,0)(0.7,0)(1.05,0)(1.4,0)(1.75,0)(2.1,0)(2.45,0)(2.8,0)(3.15,0)};
  \draw[dotted,thick](0.25,-0.1)--(0.25,0.1)--(2.2,0.1)--(2.2,-0.1)--cycle;
\end{tikzpicture} :=
{} - \begin{tikzpicture}
  \draw[-] (0,0)--(3.15,0);
  \draw plot[only marks, mark=*, mark options = {color=black, scale=.5}] coordinates {(0,0)(0.35,0)(0.7,0)(1.05,0)(1.4,0)(1.75,0)(2.1,0)(2.45,0)(2.8,0)(3.15,0)};
  \draw[-](0.35,0.1)--(2.1,0.1);
  \draw[-](0.35,-0.1)--(2.1,-0.1);
  \draw (0.35,0.1) arc[radius = 0.1,start angle= 90,end angle = 270];
  \draw (2.1,0.1) arc[radius = 0.1,start angle= 90,end angle = -90];
\end{tikzpicture}, \\
& \begin{tikzpicture}
  \draw[-] (0,0)--(3.85,0);
  \draw plot[only marks, mark=*, mark options = {color=black, scale=.5}] coordinates {(0,0)(0.35,0)(0.7,0)(1.05,0)(1.4,0)(1.75,0)(2.1,0)(2.45,0)(2.8,0)(3.15,0)(3.5,0)(3.85,0)};
  \draw[dotted,thick](0.25,-0.1)--(0.25,0.1)--(3.6,0.1)--(3.6,-0.1)--cycle;
\end{tikzpicture} :=
\begin{tikzpicture}
  \draw[-] (0,0)--(3.85,0);
  \draw plot[only marks, mark=*, mark options = {color=black, scale=.5}] coordinates {(0,0)(0.35,0)(0.7,0)(1.05,0)(1.4,0)(1.75,0)(2.1,0)(2.45,0)(2.8,0)(3.15,0)(3.5,0)(3.85,0)};
  \draw[-](0.35,0.1)--(1.4,0.1);
  \draw[-](0.35,-0.1)--(1.4,-0.1);
  \draw (0.35,0.1) arc[radius = 0.1,start angle= 90,end angle = 270];
  \draw (1.4,0.1) arc[radius = 0.1,start angle= 90,end angle = -90];
  \draw[-](1.75,0.1)--(3.5,0.1);
  \draw[-](1.75,-0.1)--(3.5,-0.1);
  \draw (1.75,0.1) arc[radius = 0.1,start angle= 90,end angle = 270];
  \draw (3.5,0.1) arc[radius = 0.1,start angle= 90,end angle = -90];
\end{tikzpicture} +
\begin{tikzpicture}
  \draw[-] (0,0)--(3.85,0);
  \draw plot[only marks, mark=*, mark options = {color=black, scale=.5}] coordinates {(0,0)(0.35,0)(0.7,0)(1.05,0)(1.4,0)(1.75,0)(2.1,0)(2.45,0)(2.8,0)(3.15,0)(3.5,0)(3.85,0)};
  \draw[-](0.35,0.1)--(2.1,0.1);
  \draw[-](0.35,-0.1)--(2.1,-0.1);
  \draw (0.35,0.1) arc[radius = 0.1,start angle= 90,end angle = 270];
  \draw (2.1,0.1) arc[radius = 0.1,start angle= 90,end angle = -90];
  \draw[-](2.45,0.1)--(3.5,0.1);
  \draw[-](2.45,-0.1)--(3.5,-0.1);
  \draw (2.45,0.1) arc[radius = 0.1,start angle= 90,end angle = 270];
  \draw (3.5,0.1) arc[radius = 0.1,start angle= 90,end angle = -90];
\end{tikzpicture} \\
& \hspace{240.5pt} -
\begin{tikzpicture}
  \draw[-] (0,0)--(3.85,0);
  \draw plot[only marks, mark=*, mark options = {color=black, scale=.5}] coordinates {(0,0)(0.35,0)(0.7,0)(1.05,0)(1.4,0)(1.75,0)(2.1,0)(2.45,0)(2.8,0)(3.15,0)(3.5,0)(3.85,0)};
  \draw[-](0.35,0.1)--(3.5,0.1);
  \draw[-](0.35,-0.1)--(3.5,-0.1);
  \draw (0.35,0.1) arc[radius = 0.1,start angle= 90,end angle = 270];
  \draw (3.5,0.1) arc[radius = 0.1,start angle= 90,end angle = -90];
\end{tikzpicture}~.
\end{aligned}
\end{equation*}%
With the first two diagrams in the last line of \eqref{diagram example} replaced by \eqref{dotted example 1}, the number of diagrams to be summed over can be reduced by 1 and we now compute $\Ls_b^c(s_1,\cdots,s_{10})$ as 
\begin{equation}\label{opt example}
\begin{aligned}
 \begin{tikzpicture}
  \draw[-] (0,0)--(3.15,0);
  \draw plot[only marks, mark=*, mark options = {color=black, scale=.5}] coordinates {(0,0)(0.35,0)(0.7,0)(1.05,0)(1.4,0)(1.75,0)(2.1,0)(2.45,0)(2.8,0)(3.15,0)};
  \draw[-](0,0.1)--(3.15,0.1);
  \draw[-](0,-0.1)--(3.15,-0.1);
  \draw (0,0.1) arc[radius = 0.1,start angle= 90,end angle = 270];
  \draw (3.15,0.1) arc[radius = 0.1,start angle= 90,end angle = -90];
\end{tikzpicture}
& =  
\begin{tikzpicture}
  \draw[-] (0,0)--(3.15,0);
  \draw plot[only marks, mark=*, mark options = {color=black, scale=.5}] coordinates {(0,0)(0.35,0)(0.7,0)(1.05,0)(1.4,0)(1.75,0)(2.1,0)(2.45,0)(2.8,0)(3.15,0)};
  \draw[-](-0.1,-0.1)--(-0.1,0.1)--(3.25,0.1)--(3.25,-0.1)--cycle;
\end{tikzpicture}
& \text{(zero ``\scalebox{.6}{$\bullet$}''s in
\begin{tikzpicture}
  \draw[dotted,thick](0,-0.1)--(0.5,-0.1)--(0.5,0.1)--(0,0.1)--cycle;
\end{tikzpicture})}\\
& + 
\begin{tikzpicture}
  \draw[-] (0,0)--(3.15,0);
  \draw plot[only marks, mark=*, mark options = {color=black, scale=.5}] coordinates {(0,0)(0.35,0)(0.7,0)(1.05,0)(1.4,0)(1.75,0)(2.1,0)(2.45,0)(2.8,0)(3.15,0)};
  \draw[dotted,thick](-0.1,-0.1)--(-0.1,0.1)--(1.15,0.1)--(1.15,-0.1)--cycle;
  \draw[-](1.3,-0.1)--(1.3,0.1)--(3.25,0.1)--(3.25,-0.1)--cycle;
\end{tikzpicture} 
+ \cdots \cdots  + 
\begin{tikzpicture}
  \draw[-] (0,0)--(3.15,0);
  \draw plot[only marks, mark=*, mark options = {color=black, scale=.5}] coordinates {(0,0)(0.35,0)(0.7,0)(1.05,0)(1.4,0)(1.75,0)(2.1,0)(2.45,0)(2.8,0)(3.15,0)};
  \draw[dotted,thick](1.65,-0.1)--(1.65,0.1)--(2.9,0.1)--(2.9,-0.1)--cycle;
  \draw[-](3.25,-0.1)--(3.05,-0.1)--(3.05,0.15)--(1.5,0.15)--(1.5,-0.1)--(-0.1,-0.1)--(-0.1,0.15)--(-0.1,0.2)--(3.25,0.2)--cycle;
\end{tikzpicture}
& \text{(four ``\scalebox{.6}{$\bullet$}''s in
\begin{tikzpicture}
  \draw[dotted,thick](0,-0.1)--(0.5,-0.1)--(0.5,0.1)--(0,0.1)--cycle;
\end{tikzpicture})}\\
& +  \begin{tikzpicture}
  \draw[-] (0,0)--(3.15,0);
  \draw plot[only marks, mark=*, mark options = {color=black, scale=.5}] coordinates {(0,0)(0.35,0)(0.7,0)(1.05,0)(1.4,0)(1.75,0)(2.1,0)(2.45,0)(2.8,0)(3.15,0)};
  \draw[dotted,thick](-0.1,-0.1)--(-0.1,0.1)--(1.85,0.1)--(1.85,-0.1)--cycle;
  \draw[-](2,-0.1)--(2,0.1)--(3.25,0.1)--(3.25,-0.1)--cycle;
\end{tikzpicture}  
+ \cdots \cdots
+  \begin{tikzpicture}
  \draw[-] (0,0)--(3.15,0);
  \draw plot[only marks, mark=*, mark options = {color=black, scale=.5}] coordinates {(0,0)(0.35,0)(0.7,0)(1.05,0)(1.4,0)(1.75,0)(2.1,0)(2.45,0)(2.8,0)(3.15,0)};
  \draw[dotted,thick](0.95,-0.1)--(0.95,0.1)--(2.9,0.1)--(2.9,-0.1)--cycle;
  \draw[-](3.25,-0.1)--(3.05,-0.1)--(3.05,0.15)--(0.8,0.15)--(0.8,-0.1)--(-0.1,-0.1)--(-0.1,0.15)--(-0.1,0.2)--(3.25,0.2)--cycle;
\end{tikzpicture}
& \text{(six ``\scalebox{.6}{$\bullet$}''s in
\begin{tikzpicture}
  \draw[dotted,thick](0,-0.1)--(0.5,-0.1)--(0.5,0.1)--(0,0.1)--cycle;
\end{tikzpicture})}\\
&+  
\begin{tikzpicture}
  \draw[-] (0,0)--(3.15,0);
  \draw plot[only marks, mark=*, mark options = {color=black, scale=.5}] coordinates {(0,0)(0.35,0)(0.7,0)(1.05,0)(1.4,0)(1.75,0)(2.1,0)(2.45,0)(2.8,0)(3.15,0)};
  \draw[dotted,thick](0.25,-0.1)--(0.25,0.1)--(2.9,0.1)--(2.9,-0.1)--cycle;
  \draw[-](3.25,-0.1)--(3.05,-0.1)--(3.05,0.15)--(0.1,0.15)--(0.1,-0.1)--(-0.1,-0.1)--(-0.1,0.15)--(-0.1,0.2)--(3.25,0.2)--cycle;
\end{tikzpicture} 
 +  
\begin{tikzpicture}
  \draw[-] (0,0)--(3.15,0);
  \draw plot[only marks, mark=*, mark options = {color=black, scale=.5}] coordinates {(0,0)(0.35,0)(0.7,0)(1.05,0)(1.4,0)(1.75,0)(2.1,0)(2.45,0)(2.8,0)(3.15,0)};
  \draw[dotted,thick](-0.1,-0.1)--(-0.1,0.1)--(1.15,0.1)--(1.15,-0.1)--cycle;
  \draw[dotted,thick](1.65,-0.1)--(1.65,0.1)--(2.9,0.1)--(2.9,-0.1)--cycle;
  \draw[-](3.25,-0.1)--(3.05,-0.1)--(3.05,0.15)--(1.5,0.15)--(1.5,-0.1)--(1.3,-0.1)--(1.3,0.15)--(1.3,0.2)--(3.25,0.2)--cycle;
\end{tikzpicture}
& \text{(eight ``\scalebox{.6}{$\bullet$}''s in
\begin{tikzpicture}
  \draw[dotted,thick](0,-0.1)--(0.5,-0.1)--(0.5,0.1)--(0,0.1)--cycle;
\end{tikzpicture})}
\end{aligned}
\end{equation}%
For rounded boxes with more points, more diagrams can be reduced. If the complexity of evaluating each dotted box is not more expensive than a rounded box with the same length, a considerable reduction in total computational cost using this optimization can then be expected for a large $m$. In the following subsection, we consider an efficient approach to compute the dotted boxes.

\subsubsection{Iterative method for computing dotted boxes}\label{sec:dotted}

For any even integer $m \geq 4$, we denote each dotted box by $\Ls_b^d(s_1,\cdots,s_m)$, which is calculated in formulas as \begin{multline}\label{a def original}
 \underbrace{\begin{tikzpicture}
 \draw[-] (0,0)--(1.7,0);\draw plot[only marks,mark =*, mark options={color=black, scale=0.5}]coordinates {(0,0) (0.5,0) (1,0)(1.5,0)};
 \draw[dotted](1.7,0)--(2.3,0);
 \draw[-] (2.3,0)--(4,0);\draw plot[only marks,mark =*, mark options={color=black, scale=0.5}]coordinates { (2.5,0)(3,0)(3.5,0)(4,0) };
\draw[dotted,thick] (-0.1,0.1)--(4.1,0.1)--(4.1,-0.1)--(-0.1,-0.1)--(-0.1,0.1);
   \end{tikzpicture}}_{m\text{~points}} = \Ls_b^d(s_1,\cdots,s_m) :=
      \sum_{j=0}^{\lfloor m/4 \rfloor-1} (-1)^{j+1} \times \\
   \times \sum_{\substack{i_1=4\\i_1\text{ is even}}}^{m-4}
       \sum_{\substack{i_2=i_1+4\\ i_2\text{ is even}}}^{m-4} \cdots \sum_{\substack{i_j=i_{j-1}+4\\i_j\text{ is even}}}^{m-4}
       \Ls_b^c(s_{1},\cdots,s_{i_1})\Ls_b^c(s_{i_1 + 1},\cdots,s_{i_2})\cdots \Ls_b^c(s_{i_{j} + 1},\cdots,s_{m}).
\end{multline}
Note that the multiple summation in the second line above becomes one single rounded box $\Ls_b^c(s_1,\cdots,s_m)$ when the index $j=0$.

For a more efficient implementation, we compute a given dotted box iteratively based on the previous results of shorter dotted boxes. For example when $m=12$, we have
\begin{equation}\label{dotted example}
  \begin{split}
  \begin{tikzpicture}
  \draw[-] (0,0)--(3.85,0);
  \draw plot[only marks, mark=*, mark options = {color=black, scale=.5}] coordinates {(0,0)(0.35,0)(0.7,0)(1.05,0)(1.4,0)(1.75,0)(2.1,0)(2.45,0)(2.8,0)(3.15,0)(3.5,0)(3.85,0)};
  \draw[dotted,thick](-0.1,-0.1)--(-0.1,0.1)--(3.95,0.1)--(3.95,-0.1)--cycle;
  \end{tikzpicture} &=
  -\, \begin{tikzpicture}
  \draw[-] (0,0)--(3.85,0);
  \draw plot[only marks, mark=*, mark options = {color=black, scale=.5}] coordinates {(0,0)(0.35,0)(0.7,0)(1.05,0)(1.4,0)(1.75,0)(2.1,0)(2.45,0)(2.8,0)(3.15,0)(3.5,0)(3.85,0)};
  \draw[-](0,0.1)--(1.05,0.1);
  \draw[-](0,-0.1)--(1.05,-0.1);
  \draw (0,0.1) arc[radius = 0.1,start angle= 90,end angle = 270];
  \draw (1.05,0.1) arc[radius = 0.1,start angle= 90,end angle = -90];
  \draw[-](1.4,0.1)--(2.45,0.1);
  \draw[-](1.4,-0.1)--(2.45,-0.1);
  \draw (1.4,0.1) arc[radius = 0.1,start angle= 90,end angle = 270];
  \draw (2.45,0.1) arc[radius = 0.1,start angle= 90,end angle = -90];
  \draw[-](2.8,0.1)--(3.85,0.1);
  \draw[-](2.8,-0.1)--(3.85,-0.1);
  \draw (2.8,0.1) arc[radius = 0.1,start angle= 90,end angle = 270];
  \draw (3.85,0.1) arc[radius = 0.1,start angle= 90,end angle = -90];
  \end{tikzpicture}
  + \begin{tikzpicture}
  \draw[-] (0,0)--(3.85,0);
  \draw plot[only marks, mark=*, mark options = {color=black, scale=.5}] coordinates {(0,0)(0.35,0)(0.7,0)(1.05,0)(1.4,0)(1.75,0)(2.1,0)(2.45,0)(2.8,0)(3.15,0)(3.5,0)(3.85,0)};
  \draw[-](0,0.1)--(2.45,0.1);
  \draw[-](0,-0.1)--(2.45,-0.1);
  \draw (0,0.1) arc[radius = 0.1,start angle= 90,end angle = 270];
  \draw (2.45,0.1) arc[radius = 0.1,start angle= 90,end angle = -90];
  \draw[-](2.8,0.1)--(3.85,0.1);
  \draw[-](2.8,-0.1)--(3.85,-0.1);
  \draw (2.8,0.1) arc[radius = 0.1,start angle= 90,end angle = 270];
  \draw (3.85,0.1) arc[radius = 0.1,start angle= 90,end angle = -90];
  \end{tikzpicture} \\
  &\quad + \begin{tikzpicture}
  \draw[-] (0,0)--(3.85,0);
  \draw plot[only marks, mark=*, mark options = {color=black, scale=.5}] coordinates {(0,0)(0.35,0)(0.7,0)(1.05,0)(1.4,0)(1.75,0)(2.1,0)(2.45,0)(2.8,0)(3.15,0)(3.5,0)(3.85,0)};
  \draw[-](0,0.1)--(1.75,0.1);
  \draw[-](0,-0.1)--(1.75,-0.1);
  \draw (0,0.1) arc[radius = 0.1,start angle= 90,end angle = 270];
  \draw (1.75,0.1) arc[radius = 0.1,start angle= 90,end angle = -90];
  \draw[-](2.1,0.1)--(3.85,0.1);
  \draw[-](2.1,-0.1)--(3.85,-0.1);
  \draw (2.1,0.1) arc[radius = 0.1,start angle= 90,end angle = 270];
  \draw (3.85,0.1) arc[radius = 0.1,start angle= 90,end angle = -90];
  \end{tikzpicture}
  + \begin{tikzpicture}
  \draw[-] (0,0)--(3.85,0);
  \draw plot[only marks, mark=*, mark options = {color=black, scale=.5}] coordinates {(0,0)(0.35,0)(0.7,0)(1.05,0)(1.4,0)(1.75,0)(2.1,0)(2.45,0)(2.8,0)(3.15,0)(3.5,0)(3.85,0)};
  \draw[-](0,0.1)--(1.05,0.1);
  \draw[-](0,-0.1)--(1.05,-0.1);
  \draw (0,0.1) arc[radius = 0.1,start angle= 90,end angle = 270];
  \draw (1.05,0.1) arc[radius = 0.1,start angle= 90,end angle = -90];
  \draw[-](1.4,0.1)--(3.85,0.1);
  \draw[-](1.4,-0.1)--(3.85,-0.1);
  \draw (1.4,0.1) arc[radius = 0.1,start angle= 90,end angle = 270];
  \draw (3.85,0.1) arc[radius = 0.1,start angle= 90,end angle = -90];
  \end{tikzpicture} \\
  &\quad - \begin{tikzpicture}
  \draw[-] (0,0)--(3.85,0);
  \draw plot[only marks, mark=*, mark options = {color=black, scale=.5}] coordinates {(0,0)(0.35,0)(0.7,0)(1.05,0)(1.4,0)(1.75,0)(2.1,0)(2.45,0)(2.8,0)(3.15,0)(3.5,0)(3.85,0)};
  \draw[-](0,0.1)--(3.85,0.1);
  \draw[-](0,-0.1)--(3.85,-0.1);
  \draw (0,0.1) arc[radius = 0.1,start angle= 90,end angle = 270];
  \draw (3.85,0.1) arc[radius = 0.1,start angle= 90,end angle = -90];
  \end{tikzpicture} \\
  &= -\, \begin{tikzpicture}
  \draw[-] (0,0)--(3.85,0);
  \draw plot[only marks, mark=*, mark options = {color=black, scale=.5}] coordinates {(0,0)(0.35,0)(0.7,0)(1.05,0)(1.4,0)(1.75,0)(2.1,0)(2.45,0)(2.8,0)(3.15,0)(3.5,0)(3.85,0)};
  \draw[dotted,thick](-0.1,-0.1)--(-0.1,0.1)--(2.55,0.1)--(2.55,-0.1)--cycle;
  \draw[-](2.8,0.1)--(3.85,0.1);
  \draw[-](2.8,-0.1)--(3.85,-0.1);
  \draw (2.8,0.1) arc[radius = 0.1,start angle= 90,end angle = 270];
  \draw (3.85,0.1) arc[radius = 0.1,start angle= 90,end angle = -90];
  \end{tikzpicture}
  - \begin{tikzpicture}
  \draw[-] (0,0)--(3.85,0);
  \draw plot[only marks, mark=*, mark options = {color=black, scale=.5}] coordinates {(0,0)(0.35,0)(0.7,0)(1.05,0)(1.4,0)(1.75,0)(2.1,0)(2.45,0)(2.8,0)(3.15,0)(3.5,0)(3.85,0)};
  \draw[dotted,thick](-0.1,-0.1)--(-0.1,0.1)--(1.85,0.1)--(1.85,-0.1)--cycle;
  \draw[-](2.1,0.1)--(3.85,0.1);
  \draw[-](2.1,-0.1)--(3.85,-0.1);
  \draw (2.1,0.1) arc[radius = 0.1,start angle= 90,end angle = 270];
  \draw (3.85,0.1) arc[radius = 0.1,start angle= 90,end angle = -90];
  \end{tikzpicture} \\
  & \quad - \begin{tikzpicture}
  \draw[-] (0,0)--(3.85,0);
  \draw plot[only marks, mark=*, mark options = {color=black, scale=.5}] coordinates {(0,0)(0.35,0)(0.7,0)(1.05,0)(1.4,0)(1.75,0)(2.1,0)(2.45,0)(2.8,0)(3.15,0)(3.5,0)(3.85,0)};
  \draw[dotted,thick](-0.1,-0.1)--(-0.1,0.1)--(1.15,0.1)--(1.15,-0.1)--cycle;
  \draw[-](1.4,0.1)--(3.85,0.1);
  \draw[-](1.4,-0.1)--(3.85,-0.1);
  \draw (1.4,0.1) arc[radius = 0.1,start angle= 90,end angle = 270];
  \draw (3.85,0.1) arc[radius = 0.1,start angle= 90,end angle = -90];
  \end{tikzpicture}
  - \begin{tikzpicture}
  \draw[-] (0,0)--(3.85,0);
  \draw plot[only marks, mark=*, mark options = {color=black, scale=.5}] coordinates {(0,0)(0.35,0)(0.7,0)(1.05,0)(1.4,0)(1.75,0)(2.1,0)(2.45,0)(2.8,0)(3.15,0)(3.5,0)(3.85,0)};
  \draw[-](0,0.1)--(3.85,0.1);
  \draw[-](0,-0.1)--(3.85,-0.1);
  \draw (0,0.1) arc[radius = 0.1,start angle= 90,end angle = 270];
  \draw (3.85,0.1) arc[radius = 0.1,start angle= 90,end angle = -90];
  \end{tikzpicture} ~,
  \end{split}
\end{equation}%
where one can see that the computation requires only 3 multiplications and 3 subtractions. More generally, such iteration is described by the lemma below:  
\begin{lemma}\label{lemma:dotted iter}
Given the increasing time sequence $(s_1,\cdots,s_m)$ with $m \geq 4$ being an even number, we have 
\begin{equation*}
\begin{split}
&\underbrace{\begin{tikzpicture}
 \draw[-] (0,0)--(1.8,0);\draw plot[only marks,mark =*, mark options={color=black, scale=0.5}]coordinates {(0,0)(0.4,0)(0.8,0)(1.2,0)(1.6,0)};
 \draw[dotted](1.8,0)--(2.4,0);
 \draw[-] (2.4,0)--(4.2,0);\draw plot[only marks,mark =*, mark options={color=black, scale=0.5}]coordinates { (2.6,0)(3,0)(3.4,0)(3.8,0)(4.2,0) };
\draw[dotted,thick] (-0.1,0.1)--(4.3,0.1)--(4.3,-0.1)--(-0.1,-0.1)--(-0.1,0.1);
   \end{tikzpicture}}_{m\text{~points}}  = 
 {} - \begin{tikzpicture}
 \draw[-] (0,0)--(1.8,0);\draw plot[only marks,mark =*, mark options={color=black, scale=0.5}]coordinates {(0,0)(0.4,0)(0.8,0)(1.2,0)(1.6,0)};
 \draw[dotted](1.8,0)--(2.4,0);
 \draw[-] (2.4,0)--(4.2,0);\draw plot[only marks,mark =*, mark options={color=black, scale=0.5}]coordinates { (2.6,0)(3,0)(3.4,0)(3.8,0)(4.2,0) };
 \draw (0,0.1) arc[radius = 0.1,start angle= 90,end angle = 270];
  \draw (4.2,0.1) arc[radius=0.1,start angle=90,end angle=-90];
  \draw[-](0,0.1)--(4.2,0.1);\draw[-](0,-0.1)--(4.2,-0.1);
   \end{tikzpicture}   
  - \begin{tikzpicture}
 \draw[-] (0,0)--(1.8,0);\draw plot[only marks,mark =*, mark options={color=black, scale=0.5}]coordinates {(0,0)(0.4,0)(0.8,0)(1.2,0)(1.6,0)};
 \draw[dotted](1.8,0)--(2.4,0);
 \draw[-] (2.4,0)--(4.2,0);\draw plot[only marks,mark =*, mark options={color=black, scale=0.5}]coordinates {(2.6,0)(3,0)(3.4,0)(3.8,0)(4.2,0) };
\draw[dotted,thick] (-0.1,0.1)--(1.3,0.1)--(1.3,-0.1)--(-0.1,-0.1)--(-0.1,0.1);
 \draw (1.6,0.1) arc[radius = 0.1,start angle= 90,end angle = 270];
  \draw (4.2,0.1) arc[radius=0.1,start angle=90,end angle=-90];
  \draw[-](1.6,0.1)--(4.2,0.1);\draw[-](1.6,-0.1)--(4.2,-0.1);
 \end{tikzpicture} \\ 
 & -  \begin{tikzpicture}
 \draw[-] (0,0)--(2.6,0);\draw plot[only marks,mark =*, mark options={color=black, scale=0.5}]coordinates {(0,0)(0.4,0)(0.8,0)(1.2,0)(1.6,0)(2,0)(2.4,0)};
 \draw[dotted](2.6,0)--(3.2,0);
 \draw[-] (3.2,0)--(4.2,0);\draw plot[only marks,mark =*, mark options={color=black, scale=0.5}]coordinates {(3.4,0)(3.8,0)(4.2,0) };
\draw[dotted,thick] (-0.1,0.1)--(2.1,0.1)--(2.1,-0.1)--(-0.1,-0.1)--(-0.1,0.1);
 \draw (2.4,0.1) arc[radius = 0.1,start angle= 90,end angle = 270];
  \draw (4.2,0.1) arc[radius=0.1,start angle=90,end angle=-90];
  \draw[-](2.4,0.1)--(4.2,0.1);\draw[-](2.4,-0.1)--(4.2,-0.1);
 \end{tikzpicture}
 - \cdots  - 
  \begin{tikzpicture}
 \draw[-] (0,0)--(1,0);\draw plot[only marks,mark =*, mark options={color=black, scale=0.5}]coordinates {(0,0)(0.4,0)(0.8,0)};
 \draw[dotted](1,0)--(1.6,0);
 \draw[-] (1.6,0)--(4.2,0);\draw plot[only marks,mark =*, mark options={color=black, scale=0.5}]coordinates { (1.8,0)(2.2,0)(2.6,0)(3,0)(3.4,0)(3.8,0)(4.2,0) };
\draw[dotted,thick] (-0.1,0.1)--(1.9,0.1)--(1.9,-0.1)--(-0.1,-0.1)--(-0.1,0.1);
 \draw (2.2,0.1) arc[radius = 0.1,start angle= 90,end angle = 270];
  \draw (4.2,0.1) arc[radius=0.1,start angle=90,end angle=-90];
  \draw[-](2.2,0.1)--(4.2,0.1);\draw[-](2.2,-0.1)--(4.2,-0.1);
 \end{tikzpicture}    
 -  
 \begin{tikzpicture}
 \draw[-] (0,0)--(1.8,0);\draw plot[only marks,mark =*, mark options={color=black, scale=0.5}]coordinates {(0,0)(0.4,0)(0.8,0)(1.2,0)(1.6,0)};
 \draw[dotted](1.8,0)--(2.4,0);
 \draw[-] (2.4,0)--(4.2,0);\draw plot[only marks,mark =*, mark options={color=black, scale=0.5}]coordinates { (2.6,0)(3,0)(3.4,0)(3.8,0)(4.2,0) };
\draw[dotted,thick] (-0.1,0.1)--(2.7,0.1)--(2.7,-0.1)--(-0.1,-0.1)--(-0.1,0.1);
 \draw (3,0.1) arc[radius = 0.1,start angle= 90,end angle = 270];
  \draw (4.2,0.1) arc[radius=0.1,start angle=90,end angle=-90];
  \draw[-](3,0.1)--(4.2,0.1);\draw[-](3,-0.1)--(4.2,-0.1);
 \end{tikzpicture} \ , 
\end{split}
\end{equation*}%
or in formulas,
\begin{equation}\label{a def iter}
   \Ls_b^d(s_1,\cdots,s_m) = -\Ls_b^c(s_1, \cdots, s_m) - \sum_{k=2}^{m/2-2}  \Ls_b^d(s_1,\cdots,s_{2k}) \Ls_b^c(s_{2k+1},\cdots,s_{m}).
\end{equation}
\end{lemma}
\begin{proof}
We consider the following resummation of the original definition \eqref{a def original} by restricting the ``length" of the last rounded box in the multiple summation, i.e., the term $\Ls_b^c(s_{i_{j} + 1},\cdots,s_{m})$:    
\begin{multline*}
    \Ls_b^d(s_1,\cdots,s_m) = - \Ls_b^c(s_1,\cdots,s_m) +  \underbrace{\Ls_b^c(s_5,\cdots,s_m)}_{\text{length}=m-4,\ i_j = 4} (-1)^2  \Ls_b^c(s_1,s_2,s_3,s_4) + \cdots \cdots + \\
        + \underbrace{\Ls_b^c(s_{m-5},\cdots,s_m)}_{\text{length}=6,\ i_j = m-6}  \sum_{j=1}^{\lfloor \frac{m-6}{4} \rfloor}   (-1)^{j+1}  \sum_{\substack{i_1=4\\i_1\text{ is even}}}^{m-4}
        \cdots \sum_{\substack{i_{j-1}=i_{j-2}+4\\i_{j-1}\text{ is even}}}^{m-4}     \Ls_b^c(s_1,\cdots,s_{i_1})\cdots \Ls_b^c(s_{i_{j-1} + 1},\cdots,s_{m-6})\\
       + \underbrace{\Ls_b^c(s_{m-3},\cdots,s_m)}_{\text{length}=4,\ i_j = m-4}  \sum_{j=1}^{\lfloor \frac{m-4}{4} \rfloor}   (-1)^{j+1}  \sum_{\substack{i_1=4\\i_1\text{ is even}}}^{m-4}
        \cdots \sum_{\substack{i_{j-1}=i_{j-2}+4\\i_{j-1}\text{ is even}}}^{m-4}     \Ls_b^c(s_1,\cdots,s_{i_1})\cdots \Ls_b^c(s_{i_{j-1} + 1},\cdots,s_{m-4}).
\end{multline*}
By decreasing the index $j$ in each term by 1, one may easily check these multiple summations will coincide the definition \eqref{a def original} for shorter dotted boxes, i.e., 
\begin{multline*}
   \Ls_b^d(s_1,\cdots,s_m) = - \Ls_b^c(s_1,\cdots,s_m)  - \Ls_b^c(s_5,\cdots,s_m) \Ls_b^d(s_1,s_2,s_3,s_4) - \cdots \cdots - \\
   - \Ls_b^c(s_{m-5},\cdots,s_m) \Ls_b^d(s_1,\cdots,s_{m-6}) - \Ls_b^c(s_{m-3},\cdots,s_m) \Ls_b^d(s_1,\cdots,s_{m-4}) ,
\end{multline*}
which proves \eqref{a def iter}.
\end{proof}

By comparing the number of diagrams that need to be summed up for a rounded box in the example \eqref{diagram example} with that for a dotted box in \eqref{dotted example}, one can easily see that computing a dotted box for a large $m$ using the above iterative method is even cheaper than computing a rounded box of the same size. Later in Section \ref{sec:complexity analysis}, we will carry out a complexity analysis on the computational cost of these dotted boxes in the entire algorithm. Now we are ready to formulate the expansion of rounded boxes $\Ls_b^c(s_1,\cdots,s_m)$ for an arbitrary even $m$ using only dotted and rectangular boxes, and propose an optimized algorithm to compute the rounded boxes.

\subsubsection{Inclusion-exclusion principle with optimization for computing rounded boxes}

The example in the previous section suggests that a rounded box is computed by the summation of all possible diagrams filled up by the nonadjacent dotted boxes covering at least four points excluding the right end and a rectangular box covering rest of the time points. The formula is provided in the following theorem:
\begin{theorem} \label{thm:optimization}
Given the increasing time sequence $(s_1,\cdots,s_m)$ with $m \geq 4$ being an even number, we have 
\begin{equation}\label{eq: inclu exclu c opt}
 \begin{split}
& \Ls_b^c(s_1,\cdots,s_m) 
 = \Ls_b^*(s_1,\cdots,s_m)  + \sum_{\substack{1 \le i_1 < j_1 \le m \\ j_1 - i_1  \geq 4 \ \mathrm{and \ is \ even}}}\Ls_b^d(s_{i_1},\cdots,s_{j_1 - 1})\Ls_b^*(\text{rest of points}) \\
& \quad +  \sum_{\substack{1 \le i_1 < j_1 < i_2 < j_2 \le m \\ j_1 - i_1 \geq 4 \ \mathrm{and \ is \ even} \\ j_2 - i_2 \geq 4\  \mathrm{and \ is \ even} }}\Ls_b^d(s_{i_1},\cdots,s_{j_1 - 1})\Ls_b^d(s_{i_2},\cdots,s_{j_2 - 1})\Ls_b^*(\text{rest of points}) \\
& \quad + \cdots \\
& \quad + \sum_{\substack{1 \le i_1 < j_1 < \cdots < i_k < j_k \le m \\ j_1 - i_1 \geq 4 \ \mathrm{and \ is \ even} \\\cdots \\ j_k - i_k \geq 4\ \mathrm{and \ is \ even}}} \Ls_b^d(s_{i_1},\cdots,s_{j_1 - 1})\cdots \Ls_b^d(s_{i_k},\cdots,s_{j_k - 1})\Ls_b^*(\text{rest  of  points})\, ,
 \end{split}
\end{equation}
where $k = \lfloor \frac{m}{5} \rfloor$. The ``rest of points" denotes all time points in $(s_1,\cdots,s_m)$ which do not occur in the brackets of any $\Ls_b^d$ in the same summand.
\end{theorem}
We remark that on the right-hand side of \eqref{eq: inclu exclu c opt}, the number of dotted boxes $\mathcal{L}_b^d(\cdots)$ in each summand does not exceed $k = \lfloor \frac{m}{5} \rfloor$ since all dotted boxes are pairwise nonadjacent and each of them includes at least four points. For example, the diagrams
\begin{equation*}
\begin{tikzpicture}
  \draw[-] (0,0)--(3.15,0);
  \draw plot[only marks, mark=*, mark options = {color=black, scale=.5}] coordinates {(0,0)(0.35,0)(0.7,0)(1.05,0)(1.4,0)(1.75,0)(2.1,0)(2.45,0)(2.8,0)(3.15,0)};
  \draw[dotted,thick](0.25,-0.1)--(0.25,0.1)--(1.5,0.1)--(1.5,-0.1)--cycle;
  \draw[dotted,thick](1.65,-0.1)--(1.65,0.1)--(2.9,0.1)--(2.9,-0.1)--cycle;
  \draw[-](3.25,-0.1)--(3.05,-0.1)--(3.05,0.15)--(0.1,0.15)--(0.1,-0.1)--(-0.1,-0.1)--(-0.1,0.15)--(-0.1,0.2)--(3.25,0.2)--cycle;
\end{tikzpicture}
  \qquad \text{or} \qquad
\begin{tikzpicture}
  \draw[-] (0,0)--(3.15,0);
  \draw plot[only marks, mark=*, mark options = {color=black, scale=.5}] coordinates {(0,0)(0.35,0)(0.7,0)(1.05,0)(1.4,0)(1.75,0)(2.1,0)(2.45,0)(2.8,0)(3.15,0)};
  \draw[dotted,thick](0.25,-0.1)--(0.25,0.1)--(0.8,0.1)--(0.8,-0.1)--cycle;
  \draw[dotted,thick](1.65,-0.1)--(1.65,0.1)--(2.9,0.1)--(2.9,-0.1)--cycle;
  \draw[-](3.25,-0.1)--(3.05,-0.1)--(3.05,0.15)--(1.5,0.15)--(1.5,-0.1)--(0.95,-0.1)--(0.95,0.15)--(0.1,0.15)--(0.1,-0.1)--(-0.1,-0.1)--(-0.1,0.15)--(-0.1,0.2)--(3.25,0.2)--cycle;
\end{tikzpicture}
\end{equation*}
are not allowed.

Now we arrive at an optimized algorithm based on inclusion-exclusion principle to calculate a rounded box $\Ls_b^c(s_1,\cdots,s_m)$: One writes the rounded box as the expansion \eqref{opt example} using Theorem \ref{thm:optimization} and then apply Algorithm \ref{algo:rectangular box} to calculate the rectangular part and Lemma \ref{lemma:dotted iter} for the dotted segments. To avoid repeated calculations caused by recursion, one should compute all rounded segments from short to long until the entire rounded box is obtained. Such procedures in general are described by Algorithm \ref{algo:optimization}. 

\begin{algorithm}
  \caption{Inclusion-exclusion principle with optimization}\label{algo:optimization}
  \begin{algorithmic}[1]
  \For{$i$ from $1$ to $m-4$}  \Comment{\emph{Initial setting}}
  \medskip
  \State $\Ls_b^c(s_{i},s_{i+1},s_{i+2},s_{i+3}) \gets B(s_{i},s_{i+2})B(s_{i+1},s_{i+3})$
  \medskip
  \State $\Ls_b^d(s_{i},s_{i+1},s_{i+2},s_{i+3}) \gets - B(s_{i},s_{i+2})B(s_{i+1},s_{i+3})$
  \medskip
  \EndFor
  \medskip
 \For{$n$ from $3$ to $\frac{m}{2}-1$}
     \For{$k$  from $1$ to $m-2n$} \Comment{{ \small \emph{Compute $k$th rounded and dotted segment with length $2n$}}} 
     \medskip
 \State Compute $\Ls_b^c(s_k,\cdots,s_{k+2n-1})$ according to \eqref{eq: inclu exclu c opt} where each $\Ls_b^*(\text{rest of points})$ is computed according to Algorithm \ref{algo:rectangular box}
 \medskip
 \State  Compute $\Ls_b^d(s_k,\cdots,s_{k+2n-1})$ according to \eqref{a def iter}
 \medskip
    \EndFor
    \EndFor
    \medskip
  \State Compute $\Ls_b^c(s_1,\cdots,s_m)$ according to \eqref{eq: inclu exclu c opt} \Comment{\emph{Final step}}
  \State \textbf{return} $\Ls_b^c(s_1,\cdots,s_m)$
   \end{algorithmic}
\end{algorithm}

We have now finished the implementation of inclusion-exclusion principle for computing the functional $\Ls_b^c(s_1,\cdots,s_m)$. Similar as computing the bath influence functional, inclusion-exclusion principle for the rounded box is less efficient than the direct summation of all linked diagrams for a small $m$. However, our complexity analysis in the next section will show that the new algorithm will outperform the direct method as $m$ becomes large. The central idea is that, when $m$ increases, the number of diagrams in \eqref{opt example} will grow significantly slower than double factorial (the growth rate of the number of diagrams in the direct method).

The proposed algorithm in this section can be regarded as the bosonic version of the algorithm introduced in \cite{Boag2018} for fermions. 
We have also further improved the algorithm by a more efficient scheme to compute the dotted boxes (Section \ref{sec:dotted}). 

\subsection{Complexity analysis}\label{sec:complexity analysis}
In this section, we will show that the computational complexity of Algorithm \ref{algo:optimization} based on inclusion-exclusion principle is significantly smaller than double factorial, which is the growth rate of the direct summation over all linked diagrams.

We denote the complexities of evaluating $2n-$point rounded and dotted segment respectively by $C_{\text{rd}}(2n)$ and $C_{\text{dt}}(2n)$. The total computational cost for $\Ls_b^c(s_1,\cdots,s_m)$ using Algorithm \ref{algo:optimization} can be immediately written down as 
\begin{equation}\label{opt complexity est}
    C_{\text{opt}}(m) = \sum_{n=2}^{m/2-1}(m-2n)\cdot\bigl( C_{\text{rd}}(2n) + C_{\text{dt}}(2n) \bigr) +  C_{\text{rd}}(m),
\end{equation}
where the last term above refers to the cost of the final step in Line 11 of Algorithm \ref{algo:optimization}. Based on the previous calculations on the shorter segments, the computational complexity of a dotted box $C_{\text{dt}}(2n)$ is simply given by
\begin{equation}\label{dt cost est}
  C_{\text{dt}}(2n) =  \underbrace{n-3}_{\text{additions}} + \underbrace{n-3}_{\text{multiplications}}  = O(n)
\end{equation}
according Lemma \ref{lemma:dotted iter}. Therefore, we focus on the estimation of the complexity of rounded segments $C_{\text{rd}}(2n)$. 

Inspired by the example \eqref{opt example}, the computational cost for $\Ls_b^c(s_k,\cdots,s_{k+2n-1})$ in Line 7 using Theorem \ref{thm:optimization} can be estimated as 
\begin{equation}\label{Crd est}
   C_{\text{rd}}(2n) \lesssim  \sum_{k = 0}^{n-1} a_{2n,2k} \cdot ( C_b(2n-2k) + k + 1 ).
\end{equation}
In the estimation above, $a_{2n,2k}$ is the number of diagrams where the total length of the dotted boxes is $2k$. $C_b(2n-2k)$ is the computational cost of a rectangular box with length $2n-2k$, which is at $O(2^{2n-2k})$ as we have discussed at the end of Section \ref{sec:rec fast algo}. $``k"$ and $``1"$ respectively counts the multiplications used among the rectangular part and dotted boxes within each diagram, and the addition between every two diagrams. For example, the last two terms in the right-hand side of \eqref{opt example} read
\begin{displaymath}
\begin{tikzpicture}
  \draw[-] (0,0)--(3.15,0);
  \draw plot[only marks, mark=*, mark options = {color=black, scale=.5}] coordinates {(0,0)(0.35,0)(0.7,0)(1.05,0)(1.4,0)(1.75,0)(2.1,0)(2.45,0)(2.8,0)(3.15,0)};
  \draw[dotted,thick](0.25,-0.1)--(0.25,0.1)--(2.9,0.1)--(2.9,-0.1)--cycle;
  \draw[-](3.25,-0.1)--(3.05,-0.1)--(3.05,0.15)--(0.1,0.15)--(0.1,-0.1)--(-0.1,-0.1)--(-0.1,0.15)--(-0.1,0.2)--(3.25,0.2)--cycle;
\end{tikzpicture} 
 +  
\begin{tikzpicture}
  \draw[-] (0,0)--(3.15,0);
  \draw plot[only marks, mark=*, mark options = {color=black, scale=.5}] coordinates {(0,0)(0.35,0)(0.7,0)(1.05,0)(1.4,0)(1.75,0)(2.1,0)(2.45,0)(2.8,0)(3.15,0)};
  \draw[dotted,thick](-0.1,-0.1)--(-0.1,0.1)--(1.15,0.1)--(1.15,-0.1)--cycle;
  \draw[dotted,thick](1.65,-0.1)--(1.65,0.1)--(2.9,0.1)--(2.9,-0.1)--cycle;
  \draw[-](3.25,-0.1)--(3.05,-0.1)--(3.05,0.15)--(1.5,0.15)--(1.5,-0.1)--(1.3,-0.1)--(1.3,0.15)--(1.3,0.2)--(3.25,0.2)--cycle;
\end{tikzpicture}
\end{displaymath}
whose computational cost contributes to the $k=4$ term in the summation \eqref{Crd est} for $C_{\text{rd}}(10)$. In each of the $a_{10,8}=2$ diagrams above, we compute a two-point rectangular box and need at most two multiplications (for the second diagram). We further claim that $a_{2n,2}=0$ since a two-point dotted box including two adjacent points always has zero value, and thus there will not exist any diagram containing a two-point dotted segment in inclusion-exclusion expansion of a given rounded box.

At this point, we only need to focus on the estimation for $a_{p,2k}$ ($p$ can be odd), which essentially is the number of nonadjacent partitions over the integers from $1$ to $p-1$ (the last point is excluded), where each dotted segment covers at least four points and the total length of all dotted segments is $2k$. The following statement provides a useful recurrence relation for the sequence $\{a_{p,2k}\}$:
\begin{lemma}
Given integers $p \geq 1$ and $0 \le k \le \lfloor \frac{p-1}{2} \rfloor$, the sequence $\{a_{p,2k}\}$ satisfies the recurrence relation
\begin{equation}\label{a recurrence}
    a_{p,2k} = a_{p-1,2k} + (a_{p-5,2k-4} + a_{p-7,2k-6} + a_{p-9,2k-8} + \cdots  + a_{p-(2k-3),4}) + 1.
\end{equation}
\end{lemma}

\begin{proof}
For simplicity, we consider the diagrams with dotted boxes only, since the rectangular boxes automatically include all the remaining points. Let the diagram
\begin{equation} \label{diag set 1}
\begin{tikzpicture}[baseline=-2,outer sep=0pt,inner sep=0pt]
  \fill[black!30!white] (-0.1,-0.1) rectangle (1.85,0.1);
  \draw[-] (0,0)--(0.55,0);
  \draw[-] (1.2,0)--(2.65,0);
  \draw[-] (3.3,0)--(3.85,0);
  \draw[dotted,thick] (0.6,0)--(1.15,0);
  \draw[dotted,thick] (2.7,0)--(3.25,0);
  \draw plot[only marks, mark=*, mark options = {color=black, scale=.5}] coordinates {(0,0)(0.35,0)(1.4,0)(1.75,0)(2.1,0)(2.45,0)(3.5,0)(3.85,0)};
  \draw[decoration={brace,mirror, amplitude=3,raise=1.5},decorate]
  (-0.1,-0.1) -- node[below=0.25] {\tiny{$2k$ points in~ \begin{tikzpicture}
    \draw[dotted, thick](-0.1,-0.1)--(-0.1,0.05)--(0.35,0.05)--(0.35,-0.1)--cycle;
  \end{tikzpicture}}} (1.85,-0.1);
\end{tikzpicture}
\end{equation}
denote the set including all diagrams with a total number of $2k$ points inside an arbitrary number of non-adjacent dotted boxes in the shaded area (each dotted box must contain at least $4$ points), and let the diagram
\begin{equation} \label{diag set 2}
\begin{tikzpicture}[baseline=-2,outer sep=0pt,inner sep=0pt]
  \fill[black!30!white] (-0.1,-0.1) rectangle (1.85,0.1);
  \draw[-] (0,0)--(0.55,0);
  \draw[-] (1.2,0)--(2.65,0);
  \draw[-] (3.3,0)--(3.85,0);
  \draw[dotted,thick] (0.6,0)--(1.15,0);
  \draw[dotted,thick] (2.7,0)--(3.25,0);
  \draw[dotted,thick] (2.35,-0.1)--(3.6,-0.1)--(3.6,0.1)--(2.35,0.1)--cycle;
  \draw plot[only marks, mark=*, mark options = {color=black, scale=.5}] coordinates {(0,0)(0.35,0)(1.4,0)(1.75,0)(2.1,0)(2.45,0)(3.5,0)(3.85,0)};
  \draw[decoration={brace,mirror, amplitude=3,raise=1.5},decorate]
  (-0.1,-0.1) -- node[below=0.25] {\tiny{$2k$ points in~ \begin{tikzpicture}
    \draw[dotted, thick](-0.1,-0.1)--(-0.1,0.05)--(0.35,0.05)--(0.35,-0.1)--cycle;
  \end{tikzpicture}}} (1.85,-0.1);
\end{tikzpicture}
\end{equation}
be the set of diagrams that add one dotted box as indicated to each of the diagrams in \eqref{diag set 1}. For example,
\begin{align*}
\begin{split}
\begin{tikzpicture}[baseline=-2,outer sep=0pt,inner sep=0pt]
  \fill[black!30!white] (-0.1,-0.1) rectangle (2.9,0.1);
  \draw[-] (0,0)--(3.85,0);
  \draw plot[only marks, mark=*, mark options = {color=black, scale=.5}] coordinates {(0,0)(0.35,0)(0.7,0)(1.05,0)(1.4,0)(1.75,0)(2.1,0)(2.45,0)(2.8,0)(3.15,0)(3.5,0)(3.85,0)};
  \draw[decoration={brace,mirror, amplitude=3,raise=1.5},decorate]
  (-0.1,-0.1) -- node[below=0.25] {\tiny{$8$ points in~ \begin{tikzpicture}
    \draw[dotted, thick](-0.1,-0.1)--(-0.1,0.05)--(0.35,0.05)--(0.35,-0.1)--cycle;
  \end{tikzpicture}}} (2.9,-0.1);
\end{tikzpicture} = \Big\{ &
\begin{tikzpicture}
  \draw[-] (0,0)--(3.85,0);
  \draw plot[only marks, mark=*, mark options = {color=black, scale=.5}] coordinates {(0,0)(0.35,0)(0.7,0)(1.05,0)(1.4,0)(1.75,0)(2.1,0)(2.45,0)(2.8,0)(3.15,0)(3.5,0)(3.85,0)};
  \draw[dotted,thick](-0.1,-0.1)--(-0.1,0.1)--(2.55,0.1)--(2.55,-0.1)--cycle;
\end{tikzpicture}\ , \\[-10pt]
& \begin{tikzpicture}
  \draw[-] (0,0)--(3.85,0);
  \draw plot[only marks, mark=*, mark options = {color=black, scale=.5}] coordinates {(0,0)(0.35,0)(0.7,0)(1.05,0)(1.4,0)(1.75,0)(2.1,0)(2.45,0)(2.8,0)(3.15,0)(3.5,0)(3.85,0)};
  \draw[dotted,thick](0.25,-0.1)--(0.25,0.1)--(2.9,0.1)--(2.9,-0.1)--cycle;
\end{tikzpicture}\ , \ 
\begin{tikzpicture}
  \draw[-] (0,0)--(3.85,0);
  \draw plot[only marks, mark=*, mark options = {color=black, scale=.5}] coordinates {(0,0)(0.35,0)(0.7,0)(1.05,0)(1.4,0)(1.75,0)(2.1,0)(2.45,0)(2.8,0)(3.15,0)(3.5,0)(3.85,0)};
  \draw[dotted,thick](-0.1,-0.1)--(-0.1,0.1)--(1.15,0.1)--(1.15,-0.1)--cycle;
  \draw[dotted,thick](1.65,-0.1)--(1.65,0.1)--(2.9,0.1)--(2.9,-0.1)--cycle;
\end{tikzpicture} \Big\}, 
\\
\begin{tikzpicture}[baseline=-2,outer sep=0pt,inner sep=0pt]
  \fill[black!30!white] (-0.1,-0.1) rectangle (1.85,0.1);
  \draw[-] (0,0)--(3.85,0);
  \draw[dotted,thick] (2.35,-0.1)--(2.35,0.1)--(3.6,0.1)--(3.6,-0.1)--cycle;
  \draw plot[only marks, mark=*, mark options = {color=black, scale=.5}] coordinates {(0,0)(0.35,0)(0.7,0)(1.05,0)(1.4,0)(1.75,0)(2.1,0)(2.45,0)(2.8,0)(3.15,0)(3.5,0)(3.85,0)};
  \draw[decoration={brace,mirror, amplitude=3,raise=1.5},decorate]
  (-0.1,-0.1) -- node[below=0.25] {\tiny{$4$ points in~ \begin{tikzpicture}
    \draw[dotted, thick](-0.1,-0.1)--(-0.1,0.05)--(0.35,0.05)--(0.35,-0.1)--cycle;
  \end{tikzpicture}}} (1.85,-0.1);
\end{tikzpicture} = \Big\{ &
\begin{tikzpicture}
  \draw[-] (0,0)--(3.85,0);
  \draw plot[only marks, mark=*, mark options = {color=black, scale=.5}] coordinates {(0,0)(0.35,0)(0.7,0)(1.05,0)(1.4,0)(1.75,0)(2.1,0)(2.45,0)(2.8,0)(3.15,0)(3.5,0)(3.85,0)};
  \draw[dotted,thick](-0.1,-0.1)--(-0.1,0.1)--(1.15,0.1)--(1.15,-0.1)--cycle;
  \draw[dotted,thick](2.35,-0.1)--(2.35,0.1)--(3.6,0.1)--(3.6,-0.1)--cycle;
\end{tikzpicture}\ , \\[-10pt]
& \begin{tikzpicture}
  \draw[-] (0,0)--(3.85,0);
  \draw plot[only marks, mark=*, mark options = {color=black, scale=.5}] coordinates {(0,0)(0.35,0)(0.7,0)(1.05,0)(1.4,0)(1.75,0)(2.1,0)(2.45,0)(2.8,0)(3.15,0)(3.5,0)(3.85,0)};
  \draw[dotted,thick](0.25,-0.1)--(0.25,0.1)--(1.5,0.1)--(1.5,-0.1)--cycle;
  \draw[dotted,thick](2.35,-0.1)--(2.35,0.1)--(3.6,0.1)--(3.6,-0.1)--cycle;
\end{tikzpicture}\ ,\ 
\begin{tikzpicture}
  \draw[-] (0,0)--(3.85,0);
  \draw plot[only marks, mark=*, mark options = {color=black, scale=.5}] coordinates {(0,0)(0.35,0)(0.7,0)(1.05,0)(1.4,0)(1.75,0)(2.1,0)(2.45,0)(2.8,0)(3.15,0)(3.5,0)(3.85,0)};
  \draw[dotted,thick](0.6,-0.1)--(0.6,0.1)--(1.85,0.1)--(1.85,-0.1)--cycle;
  \draw[dotted,thick](2.35,-0.1)--(2.35,0.1)--(3.6,0.1)--(3.6,-0.1)--cycle;
\end{tikzpicture} \Big\}.
\end{split}
\end{align*}
Then $a_{p,2k}$ is the cardinality of the diagram set \eqref{diag set 1} or \eqref{diag set 2} if there are $p-1$ points in the shaded area. Note that in the definitions of \eqref{diag set 1} and \eqref{diag set 2}, the rightmost point in any diagram is never included in the dotted boxes, and in \eqref{diag set 2}, the one point between the shaded area and the dotted box ensures that any two dotted boxes are non-adjacent. By this definition, we claim that
\begin{equation} \label{diag set decomp}
\begin{split}
\begin{tikzpicture}[baseline=-2,outer sep=0pt,inner sep=0pt]
  \fill[black!30!white] (-0.1,-0.1) rectangle (3.95,0.1);
  \draw[-] (0,0)--(0.55,0);
  \draw[-] (1.2,0)--(4.2,0);
  \draw[dotted,thick] (0.6,0)--(1.15,0);
  \draw plot[only marks, mark=*, mark options = {color=black, scale=.5}] coordinates {(0,0)(0.35,0)(1.4,0)(1.75,0)(2.1,0)(2.45,0)(2.8,0)(3.15,0)(3.5,0)(3.85,0)(4.2,0)};
  \draw[decoration={brace,mirror, amplitude=3,raise=1.5},decorate]
  (-0.1,-0.1) -- node[below=0.25] {\tiny{$2k$ points in~ \begin{tikzpicture}
    \draw[dotted, thick](-0.1,-0.1)--(-0.1,0.05)--(0.35,0.05)--(0.35,-0.1)--cycle;
  \end{tikzpicture}}} (3.95,-0.1);
\end{tikzpicture} &=
\begin{tikzpicture}[baseline=-2,outer sep=0pt,inner sep=0pt]
  \fill[black!30!white] (-0.1,-0.1) rectangle (3.6,0.1);
  \draw[-] (0,0)--(0.55,0);
  \draw[-] (1.2,0)--(4.2,0);
  \draw[dotted,thick] (0.6,0)--(1.15,0);
  \draw plot[only marks, mark=*, mark options = {color=black, scale=.5}] coordinates {(0,0)(0.35,0)(1.4,0)(1.75,0)(2.1,0)(2.45,0)(2.8,0)(3.15,0)(3.5,0)(3.85,0)(4.2,0)};
  \draw[decoration={brace,mirror, amplitude=3,raise=1.5},decorate]
  (-0.1,-0.1) -- node[below=0.25] {\tiny{$2k$ points in~ \begin{tikzpicture}
    \draw[dotted, thick](-0.1,-0.1)--(-0.1,0.05)--(0.35,0.05)--(0.35,-0.1)--cycle;
  \end{tikzpicture}}} (3.6,-0.1);
\end{tikzpicture} \cup
\begin{tikzpicture}[baseline=-2,outer sep=0pt,inner sep=0pt]
  \fill[black!30!white] (-0.1,-0.1) rectangle (2.2,0.1);
  \draw[-] (0,0)--(0.55,0);
  \draw[-] (1.2,0)--(4.2,0);
  \draw[dotted,thick] (0.6,0)--(1.15,0);
  \draw[dotted,thick] (2.7,-0.1)--(3.95,-0.1)--(3.95,0.1)--(2.7,0.1)--cycle;
  \draw plot[only marks, mark=*, mark options = {color=black, scale=.5}] coordinates {(0,0)(0.35,0)(1.4,0)(1.75,0)(2.1,0)(2.45,0)(2.8,0)(3.15,0)(3.5,0)(3.85,0)(4.2,0)};
  \draw[decoration={brace,mirror, amplitude=3,raise=1.5},decorate]
  (-0.1,-0.1) -- node[below=0.25] {\tiny{$2k{-}4$ points in~ \begin{tikzpicture}
    \draw[dotted, thick](-0.1,-0.1)--(-0.1,0.05)--(0.35,0.05)--(0.35,-0.1)--cycle;
  \end{tikzpicture}}} (2.2,-0.1);
\end{tikzpicture} \\[5pt]
& \cup \!\!\!
\begin{tikzpicture}[baseline=-2,outer sep=0pt,inner sep=0pt]
  \fill[black!30!white] (-0.1,-0.1) rectangle (1.5,0.1);
  \draw[-] (0,0)--(0.55,0);
  \draw[-] (1.2,0)--(4.2,0);
  \draw[dotted,thick] (0.6,0)--(1.15,0);
  \draw[dotted,thick] (2,-0.1)--(3.95,-0.1)--(3.95,0.1)--(2,0.1)--cycle;
  \draw plot[only marks, mark=*, mark options = {color=black, scale=.5}] coordinates {(0,0)(0.35,0)(1.4,0)(1.75,0)(2.1,0)(2.45,0)(2.8,0)(3.15,0)(3.5,0)(3.85,0)(4.2,0)};
  \draw[decoration={brace,mirror, amplitude=3,raise=1.5},decorate]
  (-0.1,-0.1) -- node[below=0.25] {\tiny{$2k{-}6$ points in~ \begin{tikzpicture}
    \draw[dotted, thick](-0.1,-0.1)--(-0.1,0.05)--(0.35,0.05)--(0.35,-0.1)--cycle;
  \end{tikzpicture}}} (1.5,-0.1);
\end{tikzpicture} \cup \cdots \cup \!\!\!
\begin{tikzpicture}[baseline=-2,outer sep=0pt,inner sep=0pt]
  \fill[black!30!white] (-0.1,-0.1) rectangle (0.96,0.1);
  \draw[-] (0,0)--(0.55,0);
  \draw[-] (1.2,0)--(4.2,0);
  \draw[dotted,thick] (0.6,0)--(1.15,0);
  \draw[dotted,thick] (1.06,0.1)--(1.06,-0.1)--(3.95,-0.1)--(3.95,0.1)--cycle;
  \draw plot[only marks, mark=*, mark options = {color=black, scale=.5}] coordinates {(0,0)(0.35,0)(1.4,0)(1.75,0)(2.1,0)(2.45,0)(2.8,0)(3.15,0)(3.5,0)(3.85,0)(4.2,0)};
  \draw[decoration={brace,mirror, amplitude=3,raise=1.5},decorate]
  (-0.1,-0.1) -- node[below=0.25] {\tiny{$4$ points in~ \begin{tikzpicture}
    \draw[dotted, thick](-0.1,-0.1)--(-0.1,0.05)--(0.35,0.05)--(0.35,-0.1)--cycle;
  \end{tikzpicture}}} (0.96,-0.1);
  \draw[decoration={brace,mirror, amplitude=3,raise=1.5,aspect=0.6},decorate]
  (1.06,-0.1) -- node[below=0.25] {\tiny{$\qquad 2k{-}4$ points}} (3.95,-0.1);
\end{tikzpicture} \\[5pt]
& \cup \bigg\{ \begin{tikzpicture}[baseline=-2,outer sep=0pt,inner sep=0pt]
  \draw[-] (0,0)--(0.55,0);
  \draw[-] (1.2,0)--(4.2,0);
  \draw[dotted,thick] (0.6,0)--(1.15,0);
  \draw[dotted,thick] (0.78,0.1)--(0.78,-0.1)--(3.95,-0.1)--(3.95,0.1)--cycle;
  \draw plot[only marks, mark=*, mark options = {color=black, scale=.5}] coordinates {(0,0)(0.35,0)(1.4,0)(1.75,0)(2.1,0)(2.45,0)(2.8,0)(3.15,0)(3.5,0)(3.85,0)(4.2,0)};
  \draw[decoration={brace,mirror, amplitude=3,raise=1.5},decorate]
  (0.78,-0.1) -- node[below=0.25] {\tiny{$2k$ points}} (3.95,-0.1);
\end{tikzpicture} ~\bigg\},
\end{split}
\end{equation}
where all diagrams have the same length $p$. Since all the sets on the right-hand side are disjoint (which can be observed by focusing on the last dotted box of the diagrams), we immediately get \eqref{a recurrence} by counting the number of diagrams on both sides.

To show \eqref{diag set decomp}, we can take any diagram on the left-hand side, and check the status of the last second point:
\begin{itemize}
\item If the last second point is not included in any dotted boxes, then the diagram must belong to the first set on the right-hand side of \eqref{diag set decomp};
\item If the last second point is included in a $2n$-point dotted box ($1 < n < k-1$), then the diagram must belong to the $n$th set on the right-hand side of \eqref{diag set decomp};
\item If the last second point is included in a $2k$-point dotted box, then this diagram must be the one in the last line of \eqref{diag set decomp}.
\end{itemize}
It is now clear that the left-hand side of \eqref{diag set decomp} is a subset of its right-hand side. For the reverse direction, the right-hand side is obviously a subset of the left-hand side since any diagram in any set on the right-hand side of \eqref{diag set decomp} has in total $2k$ points in dotted boxes, and the last point is never included into any boxes. This completes the proof of the lemma.
\end{proof}

With the recurrence relation for the sequence $\{a_{p,2k}\}$, we now consider a two-variable generating function $f(x,y)$ with a Maclaurin expansion as 
\begin{equation}
\label{eq:f}
f(x,y) = \sum^\infty_{p=0} \sum^p_{q=0} a_{p,p-q} \cdot x^{p} y^{q}
\end{equation}
where we set $a_{p',0}=1$ and $a_{p',q'}=0$ for any nonnegative $p'$ and odd $q'$. By \eqref{a recurrence}, we have 
\begin{displaymath}
   \sum^\infty_{p=1} \sum^\infty_{q=1}  \big[ a_{p,p-q} - a_{p-1,p-q} - (a_{p-5,p-q-4}+a_{p-7,p-q-6}+ \cdots  ) \big] x^{p} y^{q} = 0   
\end{displaymath}
leading to
\begin{displaymath}
(f(x,y) - a_{0,0}) - xyf(x,y) - (x^5y + x^7y + \cdots ) f(x,y) = 0 \end{displaymath}
and thus we obtain the explicit expression of the generating function:
\begin{equation} \label{eq:fxy}
f(x,y) = \frac{1}{1-(xy+x^5y\frac{1}{1-x^2})}.
\end{equation}

Now we return to the estimation for the complexity \eqref{Crd est}, which can be further bounded by 
\begin{equation}\label{Crd est 2}
     \begin{split}
         C_{\rm{rd}}(2n) \lesssim & \ \sum_{k = 0}^{n-1} a_{2n,2k} \cdot ( 2^{2n-2k} + k + 1 ) \\
         \le & \ \sum_{k = 0}^{n}  a_{2n,2n-2k} \cdot 2^{2k} + n \cdot \sum_{k = 0}^{n} a_{2n,2n-2k} .
         \end{split}
\end{equation}
By \eqref{eq:f}, we see that
\begin{displaymath}
f(x,2) = \sum_{p=0}^{\infty} 
\left( \sum_{q=0}^{p} a_{p,p-q} \cdot 2^q \right) x^p =
\sum_{p=0}^{\infty} \left( \sum_{k=0}^{\lfloor p/2 \rfloor} a_{p,p-2k} \cdot 2^{2k} \right) x^p,
\end{displaymath}
which shows that in the Maclaurin expansion of $f(x,2)$, the coefficient of $x^{2n}$ equals the the first summation in the second line of \eqref{Crd est 2}. According to \eqref{eq:fxy}, the function $f(x,2)$ is a rational function, so that its Maclaurin expansion can be found via the Heaviside cover-up method \cite{Thomas1988}:
\begin{displaymath}
f(x,2) = \frac{x^2 - 1}{2x^5 - 2x^3  + x^2 + 2x- 1} = \sum_{i = 1}^{5} \frac{c_i}{x-x_i} = \sum_{i = 1}^{5} \left( -\frac{c_i}{x_i} \right) \sum^{\infty}_{j=0} \left( \frac{1}{x_i}\right)^j x^j
\end{displaymath}
where $x_1 \approx 0.470417$, $x_{2,3} \approx -0.970009 \pm 0.4461 \ii$ and $x_{4,5} \approx 0.7348 \pm 0.62649 \ii$ are the poles of $f(x,2)$ and $c_i$ are some constants. Therefore, asymptotically we have 
\begin{displaymath}
 \sum_{k = 0}^{n}  a_{2n,2n-2k} \cdot 2^{2k} = -\sum_{i=1}^5 \frac{c_i}{x_i^{2n+1}} \sim O\left(\max_{i=1,\cdots,5} \left( \left| \frac{1}{x_i} \right|^{2n}\right)\right) = O\left(\left| \frac{1}{x_1} \right|^{2n}\right) \approx O(2.12577^{2n}).
\end{displaymath}
Similarly, the second summation in the last line of \eqref{Crd est 2} is the coefficient of $x^{2n}$ in the Maclaurin expansion of $f(x,1)$ and we can deduce that $\sum_{k = 0}^{n} a_{2n,2n-2k} \sim O(1.44327^{2n})$. Hence, 
\begin{align*}
    C_{\text{rd}}(2n)  \lesssim   2.12577^{2n} + n\cdot 1.44327^{2n}   \sim  O(\alpha^n) \text{~with~}  \alpha \approx  4.51891.
\end{align*}
Afterwards, we insert the above upper bound back into \eqref{opt complexity est} to obtain the overall complexity $C_{\mathrm{opt}}(m)$ of Algorithm \ref{algo:optimization}:
\begin{align*}
     C_{\mathrm{opt}} (m) \lesssim & \ \left( \sum_{n=2}^{m/2-1}(m-2n)\cdot( \alpha^n + n ) \right) +  \alpha^{m/2}  \\
    = &  \ \frac{ (\frac{m}{2})^3(\alpha-1)^2 - \frac{m}{2}(6\alpha^3 + \alpha^2 - 14\alpha +7) + 6(\alpha^{m/2+1}  +  \alpha^3 - \alpha^2 - 2\alpha +1  )  }{3(\alpha-1)^2} +  \alpha^{m/2} .
\end{align*}
The fraction in the last line above is asymptotically $O(\alpha^{m/2-1})$ and thus the second term dominates the upper bound. Such estimation indicates that the major computational cost of the algorithm is spent on the rectangular boxes in the final step (Line 11) when calculating the longest rounded box. To summarize the analysis, we state the conclusion in the theorem below: 
\begin{theorem}\label{thm:complexity}
Given the increasing time sequence $(s_1,\cdots,s_m)$ with $m$ being an even number, the complexity of Algorithm \ref{algo:optimization} computing the entire rounded box $\Ls_b^c(s_1,\cdots,s_m)$ can be bounded by
\begin{equation}\label{COpt complexity}
C_{\mathrm{opt}}(m) \lesssim \alpha^{m/2}  \text{~with~} \alpha  \approx 4.51891.
\end{equation}
\end{theorem}

Compared to the direct method whose computational cost grows as fast as double factorial in $m$, the inclusion-exclusion principle based algorithm with exponential growth rate is obviously more efficient when $m$ is large.

\section{Numerical experiments}\label{sec:num exp}
In this section, we will first numerically verify the statements on the complexities of algorithms, and then simulate both bare dQMC and inchworm Monte Carlo method to see how we can benefit from the inclusion-exclusion principle in applications.

We consider the spin-boson model \cite{Wang2000, Kernan2002, Duan2017} where the Hamiltonian and perturbation operators associated to the system are 
\begin{displaymath}
H_s = \epsilon \hat{\sigma}_z + \Delta \hat{\sigma}_x, \qquad W_s = \hat{\sigma}_z
\end{displaymath}
where $\hat{\sigma}_x$ and $\hat{\sigma}_z$ are the usual Pauli matrices
\begin{displaymath}
    \hat{\sigma}_x = \begin{pmatrix}
    0 & 1\\
    1 & 0
    \end{pmatrix},\qquad
     \hat{\sigma}_z = \begin{pmatrix}
    1 & 0\\
    0 & -1
    \end{pmatrix}.
\end{displaymath}
The observable of interest is set to be $O = \hat{\sigma}_z \otimes \mathrm{Id}_b$, which meets the condition that $O$ only acts on the system space. The initial density matrix $\rho = \rho_s \otimes \rho_b$ is given by 
\begin{displaymath}
   \rho_s = \begin{pmatrix} 1 & 0 \\ 0 & 0 \end{pmatrix} \quad  \text{~and~} \quad  \rho_b = Z^{-1} \exp(-\beta H_b)\,,
\end{displaymath}
where $Z$ is a normalizing factor chosen such that $\tr(\rho_b) = 1$.

Assume a Ohmic spectral density, the two-point correlation function is formulated as 
\begin{displaymath}
B(\tau_1, \tau_2) = \sum_{l=1}^L \frac{c_l^2}{2\omega_l} \left[
  \coth \left( \frac{\beta \omega_l}{2} \right) \cos \big( \omega_l \Delta \tau \big)
  - \ii \sin\big( \omega_l \Delta \tau )
\right],
\end{displaymath} 
where $\Delta \tau$ is the time difference on the Keldysh contour defined as
\begin{equation*}
 \Delta \tau =
  \begin{cases}
   \tau_2 - \tau_1,
    & \text{if } \tau_1 \le \tau_2 < t, \\
   \tau_1 - \tau_2,
    & \text{if } t \le \tau_1 \le \tau_2, \\
   2t - \tau_1 - \tau_2,
    & \text{if } \tau_1 < t \le \tau_2.
  \end{cases}
\end{equation*}
and the coupling intensity $c_l$ and frequency of each harmonic oscillator $\omega_l$ are given by 
\begin{displaymath}
c_l = \omega_l \sqrt{\frac{\xi \omega_c}{L} [1 - \exp(-\omega_{\max}/\omega_c)]}, \quad \omega_l = -\omega_c
  \ln \left( 1 - \frac{l}{L} [1 - \exp(-\omega_{\max} / \omega_c)] \right),
  \ l = 1,\cdots,L.
\end{displaymath}

In our experiments, we will study two examples with the parameter settings listed below in Table \ref{tab:parameters}. As one can observe from Figure \ref{fig:bath cor}, $B(\tau_1,\tau_2)$ under the two parameter settings both decay to zero for large time difference, which guarantees the convergence of the Dyson series as well as the infinite series in the inchworm integro-differential equation \eqref{eq: inchworm equation}.
In Case 2, the decay of $|B(\tau_1, \tau_2)|$ is slower than Case 1, leading to a slower convergence of the Dyson series and the inchworm series. It can then be expected that larger $m$ needs to be included in the simulation of Case 2.


  \begin{table}[!htb]  
  \centering
\caption{Parameter settings for spin-boson model.}
 \label{tab:parameters}
\begin{tabular}{l@{\hspace{10pt}} c@{\hspace{10pt}} c}
\hline
\hline
    Parameters & Case 1 & Case 2  \\
\hline
     Kondo parameter, $\xi$  & 0.4 & 0.1\\
     Inverse temperature, $\beta$ & 5 & 0.2 \\
     Primary frequency, $\omega_c$ & 2.5 & 1\\
     Maximum frequency, $\omega_{\max}$ & $4$ & $4$ \\
     Energy difference, $\epsilon$ & 1    & 1 \\
     Frequency of spin flipping, $\Delta$  & 1 & 1\\
     Number of modes, $L$ & 400 & 400\\
\hline
\hline
\end{tabular}
\end{table}

\begin{figure}[!ht]
    \centering
    \includegraphics[width=0.5\textwidth]{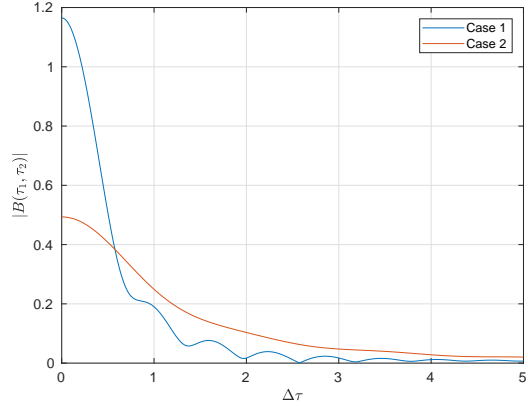}
    \caption{Two-point correlation functions under different parameter settings}
    \label{fig:bath cor}
\end{figure}

\subsection{Numerical experiments for computational complexity}
\label{sec:num exp time complex}
In this section, we compare the wall clock time on evaluating given $\Ls_b(s_1,\cdots,s_m)$ and $\Ls_b^c(s_1,\cdots,s_m)$ using direct summation and Algorithms \ref{algo:rectangular box} and \ref{algo:optimization} based on the inclusion-exclusion principle. The experiments are carried out using MATLAB on Intel Xeon CPU X5650 and the results for the time consumed may vary for different hardware, programming languages and implementation details. Since the operation counts do not depend on the value of $B(\tau_1, \tau_2)$, we will only use the parameters for Case 1 in our test throughout this section. 

\begin{figure}[!ht]
    \centering
    \includegraphics[width=\textwidth]{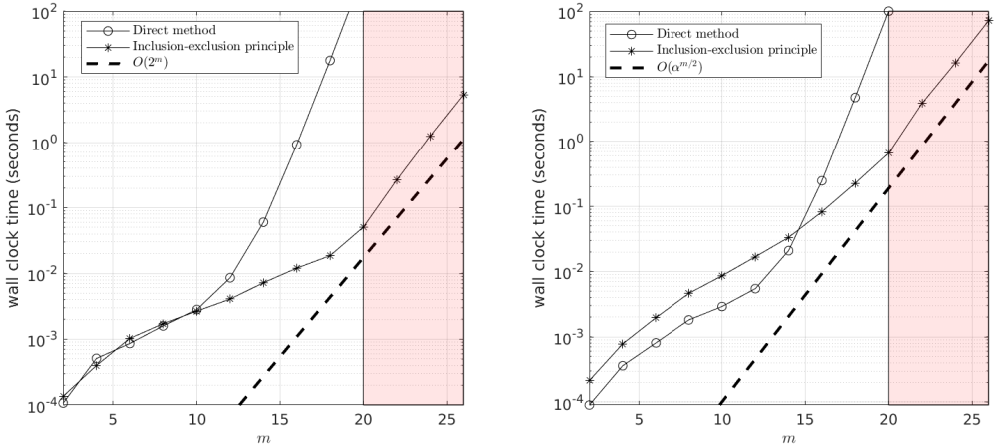}
    \caption{Wall clock time (seconds) on evaluating a given $\Ls_b(s_1,\cdots,s_m)$ (left) and $\Ls_b^c(s_1,\cdots,s_m)$ (right) using direct method and inclusion-exclusion principle in logarithm scale}
    \label{fig:algo2 vs direct}
\end{figure}

The computational time for a various choice of $m$ is plotted in Figure \ref{fig:algo2 vs direct}. In the left panel, we compare direct method with Algorithm \ref{algo:rectangular box} for computing a given bath influence functional. As predicted, the efficiencies of the two algorithms are comparable for small order $m$. Starting from $m=12$, however, due to the double factorial growth in complexity, the time cost for the direct method becomes obviously larger than the inclusion-exclusion principle whose growth rate is only exponential as $O(2^m)$ according to our discussion at the end of Section \ref{sec:Dyson}. The right panel of Figure \ref{fig:algo2 vs direct} compares Algorithm \ref{algo:optimization} with the direct summation over the linked diagrams to compute a rounded box. Algorithm \ref{algo:optimization} outperforms the direct method when $m \geq 16$. The dotted line represents our estimation of the growth rate $O(\alpha^{m/2})$. We can observe that the curve of the inclusion-exclusion principle gradually becomes parallel to the dotted line, as indicates that our estimation of the computational complexity is sharp.


We would also like to discuss the memory cost of the algorithms. The direct method is out of memory for our machine once entering the red region (i.e., $m > 20$) and thus the results are not presented. In our simulation, to implement the direct method efficiently, we first generate all the linked diagrams and store their configurations in the memory, so that the computational time presented in Figure \ref{fig:algo2 vs direct} can be minimized. To store the diagrams, we use a matrix of size $A_m \times m$ to record the indices of the time sequence in all linked diagrams. Here $A_m$ represents the number of diagrams, and $m$ denotes the length of the rounded box. For example, all the diagrams included in $\Ls_b^c(s_1,s_2,s_3,s_4,s_5,s_6)$ (see \eqref{linked pairs example}) are stored in the following $4\times 6$ matrix 
\begin{displaymath}
 \begin{pmatrix}
 1 & 3 & 2 & 5 & 4 & 6 \\
  1 & 4 & 2 & 5 & 3 & 6 \\
   1 & 4 & 2 & 6 & 3 & 5 \\
    1 & 5 & 2 & 4 & 3 & 6 
 \end{pmatrix},
\end{displaymath}
where each row of the matrix describes the pairing of time points (arcs) in one diagram on the right-hand side of \eqref{linked pairs example}. However, the size of this matrix grows quickly as $m$ increases due to the double factorial growth of the number of diagrams. For example, when the length of a rounded box reaches $m=22$, the size of matrix turns out to be $4342263000 \times 22$. Even if we use the uint8 data type in MATLAB (the smallest unsigned integer type that takes only one byte) to store the matrix, the total memory cost is around 89G, which is beyond the capacity of most machines. A workaround is to further compress the matrix using more compact storage patterns, or generate the diagrams during the summation. Both approaches will cause additional operations so that the computation may be further slowed down.

As a comparison, the major memory cost for inclusion-exclusion principle concentrates in the temporary storage of complex-valued $Q_{k_1k_2\cdots k_s}$ and $R_{k_1k_2\cdots k_s}$ appearing in Theorem \ref{thm:Lb}, which grows only as an exponential and the memory requirement is at most $16 \times \frac{2^{m+1}}{1000^3}$G (only 0.1344G for $m=22$). Here $16$ refers to the number of bytes for a double-precision complex number, and $2^{m+1}$ refers to the total number of entries in $Q_{k_1k_2\cdots k_s}$ and $R_{k_1k_2\cdots k_s}$. As a result, a longer diagram can be computed using the algorithm based on the inclusion-exclusion principle.
Such a memory issue for the direct method also exists when computing $\Ls_b(s_1,\cdots,s_m)$ since the number of diagrams given by a bath influence functional is even larger than that in a rounded box of the same size.

\subsection{Numerical simulations for open quantum systems}
We now implement several numerical simulations on the observable $\langle \hat{\sigma}_z(t) \rangle$ using bare dQMC and inchworm Monte Carlo method respectively, in which inclusion-exclusion principle will be used to evaluate large rectangular and rounded boxes. We would like to check if and how frequently we will encounter the scenarios when the order $m$ has to be chosen very large during the simulations to show the necessity of exploiting the inclusion-exclusion principle.

\subsubsection{Numerical methods}
  \label{sec:num method}
We first introduce the numerical methods for our simulations. In particular, we will discuss the details of the implementation of bare dQMC. For the more complicated inchworm equation, we only provide our Monte Carlo sampling method, which is novel in this work. One may refer to \cite{Cai2020b,Cai2020} for the general framework of the full implementation.

In the Dyson series \eqref{eq:observable1}, $m$ should be chosen as even due to the Wick's theorem for the bath influence functional. Moreover, the $m=0$ term in the Dyson series does not contain any time points and thus no Monte Carlo sampling is needed when applying bare dQMC. Therefore, we may take out this term and rewrite \eqref{eq:observable1} as  
\begin{equation*} 
   \begin{split} 
& \langle O(t) \rangle =   \tr\left( \rho_s \ee^{\ii t H_s } O_s \ee^{-\ii t H_s }   \right)    + \sum_{\substack{m=2 \\ m \text{~is even}}}^{+\infty}
  \ii^m \int^{2t}_0 \dd s_m  \int^{s_m}_0 \dd s_{m-1} \cdots \int^{s_2}_0  \dd s_1 \  \times \\
  & \hspace{100pt} \times   (-1)^{\#\{\sb < t\}} \tr_s(\rho_s \mathcal{U}^{(0)}(0, s_1,\cdots,s_m, 2t)) \cdot
    \mathcal{L}_b(s_1,\cdots,s_m).
   \end{split}
\end{equation*}
To approximate the infinite series in the above formula using Monte Carlo integration, we need sample
\begin{itemize}
    \item a positive even number $m$;
    \item a sequence of times: $0 \le s_1 \le s_2 \le \cdots \le s_m \le 2t$.
\end{itemize}
Once $m$ is chosen, the time sequence $(s_1, s_2, \cdots, s_m)$ can be generated by drawing a sample from the uniform distribution $U([0,2t]^m)$ and then sorting the sequence. In our previous works \cite{Cai2020b,Cai2020}, instead of sampling the even number $m$, we simply truncated the series \eqref{eq:observable1} at $m = \bar{M}$ and use the same number of samples for each $m$. In the current paper, we would propose a heuristic approach to take samples of $m$. Ideally, the probability of $m$ should be proportional to the absolute value of the integral in \eqref{eq:observable1}. Since such a function is not available, we make the following approximations:
\begin{itemize}
\item Ignore the term $\tr_s(\cdots)$ representing the system part;
\item Use the uniform distribution of $s_1, \cdots, s_m$ to represent the value of $\Ls_b(s_1, \cdots, s_m)$ in all cases.
\end{itemize}
Thus the distribution of $m$ becomes 
\begin{equation}\label{exact prob}
   \begin{split}
   \P_t(m=2M) = & \  \frac{1}{\lambda_0}  \int^{2t}_0 \dd s_{2M} \int^{s_{2M}}_0 \dd s_{2M-1} \cdots \int^{s_2}_0 \dd s_1 \  \big| \mathcal{L}_b\left(\tau,2\tau,\cdots,(2M-1)\tau,2M\tau \right) \big| \\
   = &  \ \frac{(2t)^{2M}}{\lambda_0 ( 2M)!}\cdot \big| \mathcal{L}_b\left(\tau,2\tau,\cdots,(2M-1)\tau,2M\tau\right) \big| \quad \quad \text{for~} M = 1,2,\cdots, M_{\max}
      \end{split}
\end{equation}
where $\tau = \frac{2t}{2M+1}$ and $\lambda_0$ is given such that the normalization $\sum_{M=1}^{M_{\max}}  \P_t(m=2M)  =1$ holds. Here we set $M_{\max}$ to be the maximum value of $M$ in order to prevent $m$ from being too large, which may cause unnecessary huge computational cost in the evaluation of the bath influence functional. Thereafter, the bare dQMC approximates the observable $\langle \hat{\sigma}_z(t) \rangle$ as
\begin{equation}\label{dqmc exact prob}
       \begin{split} 
& \langle \hat{\sigma}_z(t) \rangle \approx    \tr\left( \rho_s \ee^{\ii t H_s } \hat{\sigma}_z \ee^{-\ii t H_s }   \right)   + \frac{1}{N_s} \sum_{j=1}^{N_s}
 \frac{(2t)^{m^{(j)}}}{(m^{(j)})!}   \times  \left(  \P_t(m=m^{(j)})\right)^{-1} \times (-1)^{\#\{\boldsymbol{s^{(j)}} < t\}} \times \\
  & \hspace{60pt} \times   \tr_s\left(\rho_s \mathcal{U}^{(0)}(0, s^{(j)}_1,\cdots,s^{(j)}_{m^{(j)}}, 2t)\right) \cdot
    \mathcal{L}_b\left(s^{(j)}_1,\cdots,s^{(j)}_{m^{(j)}}\right) \text{~for~} m^{(j)} \sim  \text{~i.i.d~} \P_t,
   \end{split}
\end{equation}
where $N_s$ is the number of samples, and the quantities with superscript $(j)$ denote the $j$th sample.  


In order to study the evolution of the observable $\langle \hat{\sigma}_z(t) \rangle$ in the time interval $[0,T]$, we will need to compute all $\langle \hat{\sigma}_z(n h) \rangle$ for $n = 1,2,\cdots,T/h$ given the time step $h$ (the initial value is $\langle \hat{\sigma}_z(0) \rangle = \tr(\hat{\sigma}_z \rho_s )=1$ according to the definition \eqref{eq:O(t)}). Therefore, we need to first generate $\P_t$ for all $t = h, 2h, \cdots, T$, which requires the calculation of long $\Ls_b$ including up to $2M_{\max}$ time points. This can be  time-consuming when the time step $h$ is small, and it is also unnecessary for short time simulations where a large $m^{(j)}$ is unlikely to be sampled. To improve the efficiency of simulations, we consider a more accessible distribution to approximate $\P_t$. In \eqref{exact prob}, we insert the definition of the bath influence functional and reach to
\begin{equation*}
\P_t(m = 2M) =    \frac{(2t)^{2M}}{\lambda_0 ( 2M)!}\cdot  \displaystyle \sum_{\mf{q} \in \mQ(\sb)} \prod_{(s_j,s_k) \in \mf{q}} B(s_j,s_k)  =  \frac{(2\mathcal{B} t^2)^{M}}{\lambda_0 M!}, 
\end{equation*}
where the constant $\mathcal{B}$ is some average of the two-point correlation. Inspired by this formulation, we set $\mathcal{B}$ to be a constant and choose the probability mass function to be $\tP_t(m = 2M) = \lambda_1^{-1} (2\mathcal{B} t^2)^M / M!$, where $\lambda_1$ is chosen such that the normalization $\sum_{M=1}^{M_{\max}}  \tP_t(m=2M)  =1$ holds. Thus each sample $m^{(j)}$ can be drawn based on the Poisson distribution. More precisely, we sample $m$ according to
\begin{equation}\label{poisson sample}
   \frac{m^{(j)}}{2} - 1   \sim  \mathrm{Pois}(2\mathcal{B}t^2) \text{~and~}
   \left\{ \begin{array}{l l}
          \text{accept~} m^{(j)},  & \text{if~}  m^{(j)} \le 2M_{\max} ,   \\
          \text{reject~}  m^{(j)}, &  \text{if~} m^{(j)} > 2M_{\max} .   
    \end{array}\right.
\end{equation}
Note that the ``$-1$'' is needed on the left-hand side above since a standard Poisson distribution samples nonnegative integers from 0 while $M$ begins with 1.

It remains only to set a suitable value for $\mathcal{B} \in (0, \max|B|)$. Here we simply select $\mathcal{B}$ such that the probability mass function of $m$ is close to \eqref{exact prob} for $t=T$.
For example in Figure \ref{fig:p_dqmc_wc2p5}, one can compare the three probability mass functions of Poisson distributions with different $\mathcal{B}$ for the numerical example Case 1 to see that the yellow dashed-dotted line gives a satisfactory approximation to $\P_t$. Poisson distributions with some other $\mathcal{B}$ are plotted as references, which are comparatively far away from the target blue line. Therefore, we set $\mathcal{B} = 0.2$ for the Poisson distribution in Case 1.

\begin{figure}[!ht]
    \centering
    \includegraphics[width = .5\textwidth]{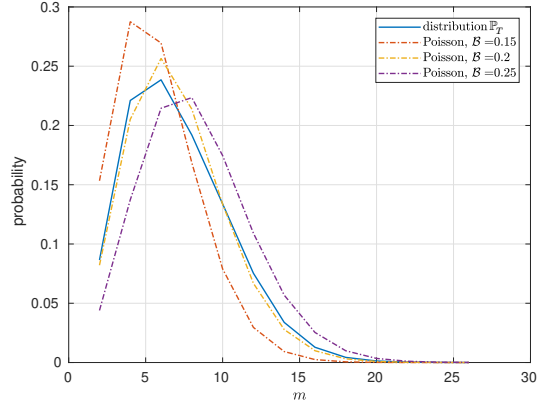} 
    \caption{Comparison between the distribution \eqref{exact prob} and Poisson distributions with changing $\mathcal{B}$ for parameter setting Case 1 with $M_{\max}=13$ and $T=2.5$.}
    \label{fig:p_dqmc_wc2p5}
\end{figure}

To end this section, we provide a brief discussion on the key procedures of the sampling method in the implementation of the inchworm Monte Carlo method. 
In \eqref{eq: inchworm equation}, the partial derivative $\partial/\partial \Sf$ on the left-hand side is discretized by a certain time integrator such as Heun's method. As for the right-hand side, similar to the bare dQMC, we need to approximate the infinite series by sampling an even number $m$ and the time sequence $(s_1,\cdots,s_{m-1})$ at every time step. The time sequence is again sampled according to the uniform distribution $U([\Si,\Sf]^m)$, and the probability mass function of $m$ is analogous to \eqref{exact prob}:
\begin{equation}\label{inchworm exact prob}
\P_t(m=2M)  =  \frac{(2t)^{2M-1}}{\lambda'_0  (2M-1)!}\cdot \big| \mathcal{L}_b^c\left(\tau, 2\tau,\cdots,(2M-1)\tau, 2M\tau\right) \big|  \ \text{for~} M = 1,2,\cdots, M_{\max},
\end{equation}
where $\tau = \frac{t}{M}$. To avoid the expensive computations of the long rounded boxes $\Ls_b^c$, we also sample each $m^{(j)}$ applying the Poisson distribution \eqref{poisson sample} in practice, where the choice of $\mathcal{B}$ is subject to a satisfactory approximation to the distribution \eqref{inchworm exact prob}, which is set to be $\mathcal{B} = 0.2$ and $\mathcal{B} = 0.3$ for Case 1 and Case 2, respectively. We refer the readers to Figure \ref{fig:p_c} for a comparison between the Poisson distribution and the distribution \eqref{inchworm exact prob} for $t = T$.

\begin{figure}[!ht]
    \centering
    \includegraphics[width = \textwidth]{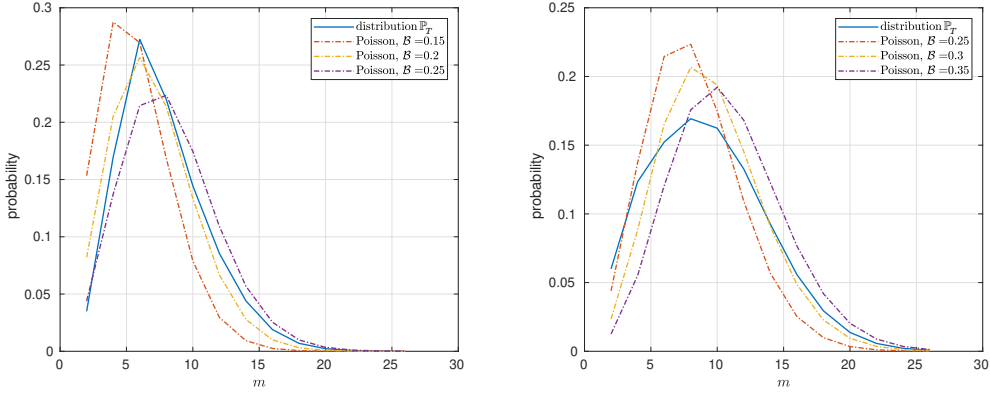}
    \caption{Comparison between the distribution \eqref{inchworm exact prob} and Poisson distributions with changing $\mathcal{B}$ for parameter setting Case 1 (left) and Case 2 (right) with $M_{\max}=13$ and $T=2.5$.}
    \label{fig:p_c}
\end{figure}

\subsubsection{Numerical results}
With the numerical methods introduced, we are now ready to present the results of our simulations on the time evolution of the observable $\langle \hat{\sigma}_z(t) \rangle$. We are particularly interested in the convergence of $\langle \hat{\sigma}_z(t) \rangle$ computed by both bare dQMC and inchworm method w.r.t.{} the order $m$. Specifically, we first perform the simulation with the series in \eqref{eq:DysonG} or \eqref{eq: inchworm equation} truncated at $m = \bar{M}$, and plot the real part of $\langle \hat{\sigma}_z(t) \rangle$ up to $T=2.5$. Note that due to the numerical error, the computed $\langle \hat{\sigma}_z(t)\rangle$ may contain a nonzero imaginary part. We hope to observe the convergence of these results to the numerical solution using our approach introduced in Section \ref{sec:num method}, which justifies our numerical method. In our simulation, the Poisson distribution is truncated at $M_{\max} = 13$, so that the maximum value of $m$ is $26$.

\begin{figure}[!ht]
    \centering
    \includegraphics[width = 0.5\textwidth]{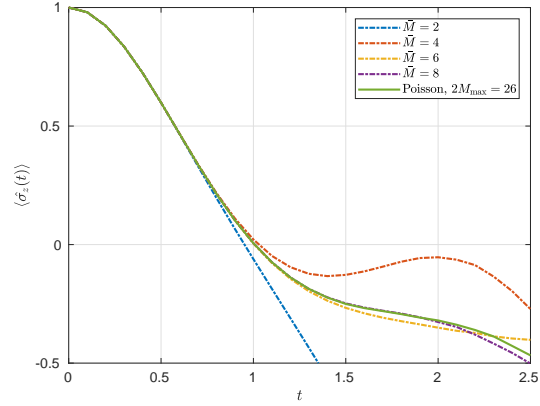}
      \caption{Evolution of $\text{Re}\langle \hat{\sigma}_z(t) \rangle$ for parameter setting Case 1 by bare dQMC}
        \label{fig:dQMC_observable}
\end{figure}

 \begin{table}[!ht]  
  \centering
\caption{Large $m$ sampled by Poisson distribution in the simulation for Case 1 by bare dQMC}
   \label{tab:m dqmc}
\begin{tabular}{c| c c c c}
\hline
  $m$ & 20 & 22 & 24 & 26  \\
\hline
     $\#\{m^{(j)} = m\}$  & 195123  &  44427 & 9197 & 1820 \\
\hline
\end{tabular}
\end{table}

Figure \ref{fig:dQMC_observable} plots the numerical results for parameter setting Case 1 using bare dQMC. We set the time step to be $h = 0.1$ and compute $\langle \hat{\sigma}_z(n h) \rangle$ for each $n = 1,\cdots,25$. For the result with $m$ sampled by Poisson distribution, each $\langle \hat{\sigma}_z(n h) \rangle$ is calculated based on $N_s = 10^8$ Monte Carlo samples. As for the results with fixed truncation $\bar{M}$, we evaluate each $m$-dimensional integral in \eqref{eq:observable1} using $N_s = 2\times 10^7$ Monte Carlo samples. One can observe that the curve of observable tends to converge as $\bar{M}$ grows. However, significant difference can still be observed between the results for $\bar{M} = 6$ and $\bar{M} = 8$, indicating that larger $m$ needs to be taken into account to get reliable results, and thus considering $m$ also as a random variable turns out to be an efficient way to find suitable number of samples. With this approach, larger $m$ will be encountered in the simulation, and we have listed in Table \ref{tab:m dqmc} the number of large $m$ (within the red region of Figure \ref{fig:algo2 vs direct}) sampled by the Poisson distribution, which also represents the number of $m$-point bath influence functionals evaluated in the entire simulation. For example, we need to compute 1820 independent $\Ls_b\bigl(s_1^{(j)},\cdots,s_{26}^{(j)}\bigr)$ for Monte Carlo integration using inclusion-exclusion principle. The evaluation of such high-order bath influence functionals is hardly feasible using the direct method. 

\begin{figure}[!ht]
    \centering
    \includegraphics[width = \textwidth]{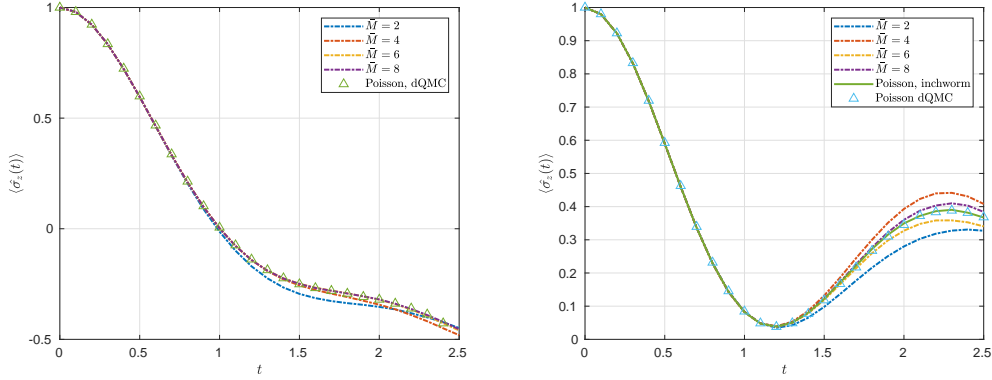}
          \caption{Evolution of $\text{Re}\langle \hat{\sigma}_z(t) \rangle$ for parameter setting Case 1 (left) and Case 2 (right) by inchworm Monte Carlo method}
    \label{fig:inchworm_observable}
\end{figure}

\begin{table}[!ht]  
  \centering
\caption{Large $m$ sampled by Poisson distribution in the simulation for Case 2 by inchworm Monte Carlo method} 
   \label{tab:m inchworm}
\begin{tabular}{c| c c c c}
\hline
    $m$ & 20 & 22 & 24 & 26  \\
\hline
   $\#\{m^{(j)} = m\}$  & 17710  &  5336 & 1492 & 395 \\
\hline
\end{tabular}
\end{table}

As for the simulations by inchworm Monte Carlo method, we refer to Figure \ref{fig:inchworm_observable} for the numerical results with both parameter settings Case 1 and Case 2. The time step is again set as $h=0.1$, while the number of samples is chosen as a relatively smaller $N_s = 10^5$ ($N_s$ denotes the total number of samples used in the simulation by Poisson distribution, and the number of samples used for each $(m-1)$-dimensional integral in the simulations with fixed truncation) since the numerical error of inchworm method is generally smaller than that of classic Dyson series \cite{Cai2020b}. For Case 1, the curve with fixed $\bar{M}$ becomes almost identical to $\bar{M}=6$ thanks to the rapid convergence of inchworm method. The result by bare dQMC using Poisson distribution (same as the solid line in Figure \ref{fig:dQMC_observable}) is given as a reference. This indicates that inchworm Monte Carlo method can provide a satisfactory approximation to the exact solution with a small truncation $\bar{M}$ and hence outperform bare dQMC for this set of parameters. As no larger $m$ is needed, there is no need to apply the inclusion-exclusion principle in this case. However in Case 2, the inchworm Monte Carlo method also suffers from slow convergence due to the slow decay of the two-point correlation function (see Figure \ref{fig:bath cor}). The right panel of Figure \ref{fig:inchworm_observable} shows that the discrepancy between $\bar{M} = 6$ and $\bar{M} = 8$ is still noticeable. With the adaptive choice of $m$, we are able to obtain results in good agreement with the reference results provide by bare dQMC with $10^8$ samples for each $\langle \hat{\sigma}_z (nh)\rangle$. Again, we list in Table \ref{tab:m inchworm} the number of samples involving large $m$ in this experiment to show that inclusion-exclusion principle is indispensable to the calculation of long rounded boxes that direct method cannot deal with.

\section{Conclusion}\label{sec:conclusion}
We have proposed fast algorithms based on inclusion-exclusion principle to sum diagrams appearing in the bare dQMC and inchworm Monte Carlo method. For bare dQMC, we have developed a formula to efficiently evaluate the bosonic bath influence functional at the cost of $O(2^m)$. Note that in the fermionic case, the bath influence functional becomes  a determinant \cite{Muhlbacher2008,Werner2006}, while in the bosonic case, the computational cost is higher, but it turns out that the computational complexity is lower than the Ryser's algorithm for matrix permanents.
For the inchworm method, our algorithm calculating the sum over linked diagrams can be considered as an extension to the work \cite{Boag2018} which deals with the fermionic quantum impurity models. By a detailed complexity analysis, we have proved that the new algorithm reduces the computational cost from the original double factorial to exponential. More precisely, we estimate the computational complexity as $O(\alpha^{m/2})$ where $\alpha \approx 4.51891$, which has been also verified by our numerical experiments. Moreover, numerical simulations for the spin-boson model have been implemented to show the advantages of our approaches.

\appendix 

\section{Formulas of functionals}
\label{app:formulas}

\subsection{Definition of $\mathcal{U}^{(0)}$}
The system related functional $\mathcal{U}^{(0)}$ in the Dyson series \eqref{eq:observable1} is defined by
\begin{displaymath}
 \mathcal{U}^{(0)}(0,s_1,\cdots,s_m ,2t)  = G_s^{(0)}(s_m, 2t) W_s G_s^{(0)}(s_{m-1}, s_{m}) W_s
  \cdots W_s G_s^{(0)}(s_1, s_2) W_s G_s^{(0)}(0, s_1),
\end{displaymath}
where 
\begin{equation*}
  G_s^{(0)}(\Sf, \Si) =
  \begin{cases}
    \ee^{-\ii (\Sf - \Si) H_s},
    & \text{if } \Si \le \Sf < t, \\[5pt]
    \ee^{-\ii (\Si - \Sf) H_s},
    & \text{if } t \le \Si \le \Sf, \\[5pt]
    \ee^{-\ii (t - \Sf) H_s} O_s \ee^{-\ii (t - \Si) H_s},
    & \text{if } \Si < t \le \Sf.
  \end{cases}
\end{equation*}

\subsection{Definition of $\mathcal{U}$}
The functional $\mc{U}$ in the integro-differential equation \eqref{eq: inchworm equation} is given by 
\begin{displaymath}
\mc{U}(\Si, s_1, \cdots, s_{m-1}, \Sf) = G(s_{m-1}, \Sf) W_s
  G(s_{m-2}, s_{m-1}) W_s \cdots W_s G(s_1, s_2) W_s G(\Si, s_1),
\end{displaymath}
where the full propagator $G(\Si,\Sf)$ is defined by
\begin{equation*}
G(\Si, \Sf) = 
\begin{cases}
  \tr_b(\rho_b G_b^{(0)}(\Sf,2t) \ee^{-\ii (\Sf - \Si) H} G_b^{(0)}( 0,\Si)),
    & \text{if } \Si \leqslant \Sf < t, \\[5pt]
  \tr_b(\rho_b G_b^{(0)}( \Sf,2t) \ee^{-\ii (\Si - \Sf) H} G_b^{(0)}(0,\Si)),
    & \text{if } t \leqslant \Si \leqslant \Sf, \\[5pt]
  \tr_b(\rho_b G_b^{(0)}(\Sf,2t) \ee^{\ii (\Sf - t) H} O
    \ee^{-\ii (t - \Si) H} G_b^{(0)}(0,\Si)),
    & \text{if } \Si < t \leqslant \Sf
\end{cases}
\end{equation*}
with the propagator associated with the bath
\begin{equation*}
  G_b^{(0)}(\Si, \Sf) = \begin{cases}
    \ee^{-\ii (\Sf - \Si) H_b},
    & \text{if } \Si \leqslant \Sf < t, \\[5pt]
    \ee^{-\ii (\Si - \Sf) H_b},
    & \text{if } t \leqslant \Si \leqslant \Sf, \\[5pt]
    \ee^{-\ii (2t - \Si - \Sf) H_b},
    & \text{if } \Si < t \leqslant \Sf.
    \end{cases}
\end{equation*}

The full propagator $G(\Si,\Sf)$ satisfies
\begin{itemize}
\item Jump condition:
\begin{equation*}  
    \begin{split}
&\lim_{\Sf \rightarrow t^+} G(\Si,\Sf) = O_s \lim_{\Sf \rightarrow t^-}G(\Si,\Sf);  \\
&\lim_{\Si \rightarrow t^-} G(\Si,\Sf) =  \lim_{\Si \rightarrow t^+}G(\Si,\Sf)O_s.
  \end{split}
  \end{equation*}
  \item Boundary condition: $G(\Sf,\Sf) = \mathrm{Id}$. 
\end{itemize}

\section{MATLAB code for computing the hafnian}
\label{app:code}
Below we provide our MATLAB code to compute the hafnian of a symmetric matrix $B$. The input matrix needs to be a symmetric square matrix with all diagonal entries being zero.
\begin{lstlisting}[language=matlab,basicstyle=\ttfamily,frame=single]
function v = hafnian(B)

m = size(B,1);
R = zeros(1,2^m-1);
Q = zeros(1,2^m); Q(1) = sum(B, 'all') / 2;
sgn = zeros(2^m,1); sgn(1) = 1;

for i=1:m
    idx = 2^(i-1);
    R(idx) = sum(B(i,:));
    for kn=1:i-1
        j = idx + 2^(kn-1);
        R(j:2*j-idx-1) = R(idx:j-1) - B(i, kn);
    end
    Q(idx+1:2*idx) = Q(1:idx) - R(idx:2*idx-1);
    sgn(idx+1:2*idx) = -sgn(1:idx);
end

v = Q.^(m/2) * sgn / factorial(m/2);
\end{lstlisting}

\bibliographystyle{plain}

\bibliography{Inchworm}

\end{document}